\documentclass[journal]{IEEEtran}
\usepackage{amsmath,amsfonts,amsthm,amssymb,mathrsfs}
\usepackage{algorithm}
\usepackage{algpseudocode}
\usepackage{array}
\usepackage{textcomp}
\usepackage{stfloats}
\usepackage{url}
\usepackage{verbatim}
\usepackage{graphicx}
\usepackage{subfigure}
\usepackage{bbding}
\usepackage{multirow}
\usepackage{bm, bbm, xcolor}
\usepackage{refcount}
\usepackage{flafter}
\usepackage{makecell}
\usepackage{diagbox}
\usepackage{booktabs}
\usepackage{graphicx}
\usepackage{cite}
\usepackage{tikz}
\usetikzlibrary{calc}

\theoremstyle{definition}

\newtheorem{proposition}{Proposition}

\newtheorem{corollary}{Corollary}
\newtheorem{remark}{Remark}
\theoremstyle{definition}

\usepackage[T1]{fontenc}
\providecommand{\url}[1]{#1}
\allowdisplaybreaks[4]

\title{Multi-Mode Pinching Antenna Systems Enabled Multi-User Communications}
\author{
Xiaoxia Xu, 
Xidong Mu,  
Yuanwei Liu, \textit{Fellow, IEEE}, 
and Arumugam Nallanathan, \textit{Fellow, IEEE}
\vspace{-0.6cm}
\thanks{X. Xu and A. Nallanathan are with the School of Electronic Engineering and Computer Science, Queen Mary University of
London, London E1 4NS, U.K. (email: \{x.xiaoxia, a.nallanathan\}@qmul.ac.uk).}
\thanks{X. Mu is with the Centre for Wireless Innovation (CWI), Queen's University Belfast, Belfast, BT3 9DT, U.K. (x.mu@qub.ac.uk)}
\thanks{Y. Liu is with the Department of Electrical and Electronic Engineering (EEE), The University of Hong Kong, Hong Kong (e-mail: yuanwei@hku.hk).}
}

\date{\today}

\begin{document}

\maketitle

\begin{abstract}
    This paper proposes a novel multi-mode pinching-antenna systems (PASS) framework. 
    Multiple data streams can be transmitted within a \emph{single} waveguide through multiple guided modes, 
    thus facilitating efficient multi-user communications through the \textit{mode-domain multiplexing}.
    A physic model is derived, which reveals the mode-selective power radiation feature of pinching antennas (PAs).
    To examine the performance, a two-mode PASS enabled two-user downlink communication system is investigated. 
    Considering the mode selectivity of PA power radiation, a practical PA grouping scheme is proposed,  
    where each PA group matches with one specific guided mode and mainly radiates its signal sequentially. 
    Depending on whether the guided mode leaks power to unmatched PAs or not, the proposed PA grouping scheme operates in either the \emph{non-leakage} or \emph{weak-leakage} regime.
    Based on this, the baseband beamforming and PA locations are jointly optimized for sum rate maximization, subject to each user's minimum rate requirement. 
    1) A simple two-PA case in the non-leakage regime is first considered. To solve the formulated problem, a channel orthogonality based solution is proposed. 
    The channel orthogonality is ensured by large-scale and wavelength-scale equality constraints on PA locations.
    Thus, the optimal beamforming reduces to maximum-ratio transmission (MRT).  
    Moreover, the optimal PA locations are obtained via a Newton-based one-dimension search algorithm that enforces two-scale PA-location constraints by Newton's method. 
    2) A general multi-PA case in both non-leakage and weak-leakage regimes is further considered. A low-complexity particle-swarm optimization with zero-forcing beamforming (PSO-ZF) algorithm is developed, thus effectively 
    tackling the high-oscillatory and strong-coupled problem. 
    Simulation results demonstrate the superiority of the proposed multi-mode PASS over conventional single-mode PASS and fixed-antenna structures.
\end{abstract}


\section{Introduction}
The ever-increasing demands of the sixth-generation (6G) wireless networks for ultra-high data rates, low latency, and massive connectivity 
have driven the development of new antenna architectures with stringent cost and energy constraints \cite{Heath2016Overview,Rappaport2013Millimeter}. 
While massive multiple-input multiple-output (MIMO) systems \cite{Marzetta2010Massive} and hybrid beamforming \cite{ElAyach2014SpatiallySparse} 
have demonstrated significant potential in millimeter-wave (mmWave) and terahertz (THz) bands, 
their implementation is hindered by hardware complexity, high power consumption, 
and limited adaptability to dynamic environments.
Recent interests in flexible and reconfigurable architectures such as intelligent reflecting surfaces (IRS) \cite{Wu2019RIS}, fluid 
antennas \cite{Wong2021Fluid}, movable antennas \cite{Zhu2024MovableAntenna}, dynamic metasurface antenna \cite{Shlezinger2021Dynamic}, 
and large intelligent surfaces (LIS) \cite{Hu2018LIS} highlight the need to move beyond traditional MIMO antenna arrays. 

The afore-mentioned systems mainly focus on customizing the scattering effects of radio environments. 
To enable flexible control over a large-scale spatial range and mitigate line-of-sight (LoS) blockages,  
\emph{pinching antenna systems} (PASS) have emerged as a novel architecture, 
which reaps both benefits of low-loss wired guided-wave communication and mobile-enabled wireless communication \cite{Liu2025PASS_Magzine,Liu2025pinching_tutorial}. 
First prototyped by NTT DOCOMO \cite{Suzuki2022Docomo}, PASS equips a lossless dielectric waveguide  
to carry electromagnetic (EM) waves over physical propagation distances that spans several meters to tens of meters. 
Multiple radiating elements, termed \emph{pinching antennas} (PAs), can be activated along the waveguide, 
thus extracting energy from the guided signal and radiate it into free space \cite{Ding2025PASS}. 
The activated PAs exchange the energy with the waveguide using EM coupling principles \cite{Wang2025Modeling,Huang1994CMT}.  
Hence, they can interact with the waveguide in a non-contact way, where the PA-waveguide spacing determines the coupling strength and radiation power \cite{Xu2025PASSPower}.  
Owing to the non-contact EM coupling nature and the operating principles, 
PAs can be easily placed at arbitrary locations along the propagation direction of the waveguide. 
This leads to several unique advantages of PASS \cite{Liu2025PASS_Magzine}. 
1) \textit{Large-scale path loss control.} 
By adjusting positions of PAs and placing PAs near the mobile user, large-scale path loss of wireless signal can be altered and reduced flexibly\cite{Xu2025Rate}.
2) \textit{Precise beamforming.} 
The flexible placement of PAs can adjust the signal phase and the array aperture, thus achieving precise spatial beamforming. 
Recent works have demonstrated the significant array gains \cite{Ouyang2025ArrayGain,Tyrovolas2025Performance} and  
high spectral efficiency \cite{Xu2025JointBF,Guo2025GPASS} of PASS.
3) \textit{Scalable architecture design.} Non-contact PAs can be easily removed and deactivated from the waveguide, 
thus enabling a scalable antenna architecture.

\subsection{Concept of Multi-Mode PASS}

\begin{figure*}[!t]
    \centering
    \includegraphics[width=0.9\linewidth]{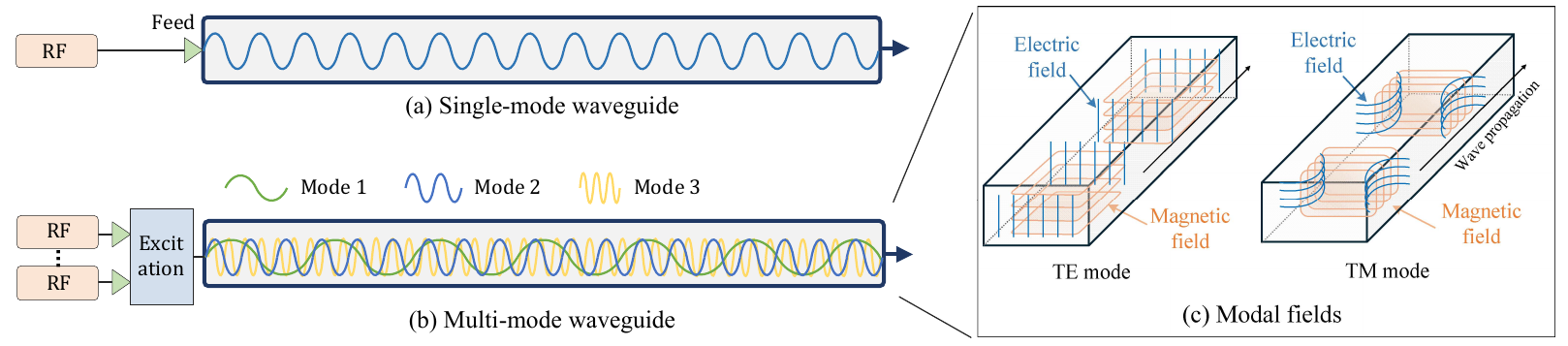}
    \caption{The concept of single-mode and multi-mode waveguides.}\label{fig:waveguide}
\end{figure*}

In existing PASS, each waveguide provides only one effective spatial dimension, leading to a rank-one spatial channel.
Hence, each waveguide supports at most one independent data stream through the spatial domain, which fundamentally constrains the system capacity and the number of served users. 
Although time-division or frequency-division multiple access can accommodate multiple users, 
these approaches suffer from reduced spectral efficiency. 
Moreover, the waveguide-division multiple access \cite{Zhao2025WDMA} and the joint pinching-and-transmit beamforming structure \cite{Liu2025PASS_Magzine} 
enable multi-user transmission by deploying multiple waveguides.
However, since connecting each additional user generally requires an extra waveguide,
the hardware cost is considerable and the resource utilization efficiency remains limited.

To overcome this degree-of-freedom (DoF) limitation, we propose a novel \emph{multi-mode PASS} framework. 
As illustrated in Fig. \ref{fig:waveguide}, the key idea is to exploit a multi-mode waveguide 
where signals over multiple guided modes can propagate simultaneously, 
thereby enabling the \emph{mode-domain multiplexing}. 
Specifically, a guided mode is a physically valid solution to Maxwell's equations under a given waveguide geometry, 
material composition, and boundary conditions.
Existing PASS typically exploits a single-mode waveguide \cite{Ding2025PASS,Liu2025PASS_Magzine,Liu2025pinching_tutorial}, 
which excites only the fundamental mode via one feed, as shown in Fig. \ref{fig:waveguide}(a). 
In contrast, the proposed multi-mode PASS excites multiple orthogonal guided modes by multiple feeds, as shown in Fig. \ref{fig:waveguide}(b).  
Each guided mode propagates with an invariant transverse field pattern, characterized by a distinct propagation constant and a unique field distribution. 
Thus, it enables each mode to carry an independent data stream with limited inter-mode interference. 
Here, Fig. \ref{fig:waveguide}(c) illustrates an example of two typical guided-mode families, namely transverse-electric (TE) and transverse-magnetic (TM) modes, 
which exhibit distinct transverse field patterns with vanishing longitudinal electric-field and magnetic-field components, respectively. 
The theoretical foundations of guided modes and modal orthogonality are well established 
in optical waveguide theory \cite{SnyderLove1983,Huang1994CMT}. 
Moreover, coupled-mode theory (CMT) suggests that efficient radiation from a desired guided mode 
requires phase matching between the PA and the corresponding modal field \cite{Huang1994CMT,Yariv1973CMT}. 
However, a physic model and systematic design methodology for multi-mode PASS are still lacking.

\subsection{Contributions}
To fill the gap, this paper proposes a novel multi-mode PASS framework. 
Multiple data streams can be efficiently transmitted using a single waveguide via the mode-domain multiplexing. 
We establish the fundamental physic model for the multi-mode PASS, which captures the mode selectivity of PA power radiation. 
We examine a two-mode PASS enabled two-user downlink communication system. 
A practical PA grouping scheme is proposed, where each PA group matches with a guided mode to sequentially radiate its signal. 
Depending on whether the guided mode leaks power to unmatched PAs or not, this scheme operates in \emph{non-leakage} and \emph{weak-leakage} regimes. 
Our main contributions can be summarized as follows. 
\begin{itemize}
    \item We propose a novel multi-mode PASS framework, where a single waveguide can simultaneously transmit multiple data streams via mode-domain multiplexing. 
    We derive a physic model to characterize the EM coupling between multi-mode waveguide and single-mode PAs, which reveals the mode selectivity of PA power radiation.  
    \item We examine the two-mode PASS performance in a downlink two-user communication system.  
    To handle the mode-selective PA power radiation, a practical PA grouping scheme is developed, where each PA group matches with a dedicated guided mode to sequentially radiate its signal.
    Based on this, we formulate the joint optimization problem of baseband beamforming and PA locations, which maximizes the system sum rate 
    subject to each user's minimum rate requirements. 
    \item We first consider a two-PA case in the non-leakage regime. To tackle the resultant nonconvex optimization problem, we propose a channel orthogonality based solution. 
    The proposed solution guarantees channel orthogonality by imposing large-scale and wavelength-scale equality conditions on PA locations. 
    Thus, the optimal beamforming reduces to maximum ratio transmission (MRT). 
    We develop a Newton-based one-dimension search algorithm, which enforces equality constraints via Newton's method and 
    explores optimal PA locations in the resulting one-dimension space.
    \item We then consider a general multi-PA case under both non-leakage and weak-leakage regimes. 
    The joint baseband beamforming and PA location optimization leads to a high-oscillatory, strongly coupled, and nonconvex problem, where many local optima exhibit.
    To avoid convergence to unfavorable local optima, we develop a low-complexity particle swarm optimization with zero-forcing beamforming (PSO-ZF) algorithm. 
    The proposed algorithm explores feasible PA positions using multiple particles, and adopts ZF beamforming for each particle to enable fast fitness evaluation. 
    \item Numerical results validate the effectivity of the proposed multi-mode PASS framework. 
    In both non-leakage and weak-leakage regimes, the proposed multi-mode PASS outperforms conventional time-division multiple access (TDMA)-based single-mode PASS 
    and fixed-antenna multiple-input single-output (MISO) structures. 
    Moreover, by initializing PSO-ZF via the optima from two-PA case, the weak-leakage regime leads to a slight performance loss compared to the non-leakage regime. 
\end{itemize}

The remainder of this paper is organized as follows. 
Section II illustrates the physic model for waveguide-to-PA radiation, and Section~III presents the signal model and formulates the joint optimization problem. 
Section~IV develops optimal solutions for a two-PA case, and Section V further investigates a general multi-PA case. 
Numerical results are presented in Section VI, and Section~VII finally concludes the paper.

\section{CMT-Based Physic Model for Multi-Mode PASS}
In this section, we first develop a physic model for a single-PA multi-mode PASS. Then, we demonstrate the mode-selectivity radiation feature and generalizes it to multi-PA case.
\subsection{Physic Model for Single PA}
We first provide a physic model for multi-mode PASS with a single PA, which captures the signal propagation within an $M$-mode waveguide. 
The signal is carried by a set of a $M$ guided modes, indexed by $\mathcal{M}=\{1,2,\dots,M\}$, which are excited by $M$ dedicated feeds at the entrance of the waveguide. 
The signal for guided mode $m$ is propagated along $x$-axis to $x_{\mathrm{PA}}$, and is then radiated by the PA located at $x_{\mathrm{PA}}$. 
The PA has a single eigenmode, which electromagnetively couples to the guided modes for signal radiation. 
The entire in-waveguide signal propagation process from feed $m$ to PA via guided mode $m$ is captured by $g_{m}$, which is defined as
\begin{equation}\label{waveguide_to_PA_single}
    g_{m} = \rho_{m}e^{-j\beta_{m}x_{\mathrm{PA}}}, ~ \forall m \in \mathcal{M},
\end{equation}
where $\rho_{m}\in\mathbb{C}$ is the power radiation coefficient for mode $m$, which characteristics both the power radiation ratio $|\rho_{m}|^2$ and the associated phase shift.
Moreover, $\beta_m$ denotes the propagation constant of guided mode $m$. 
For a dielectric waveguide, it is given by $\beta_m=n_{\mathrm{eff},m} k_0$, where $n_{\mathrm{eff},m}$ is the effective refractive index 
obtained by solving the corresponding transverse eigenvalue problem subject to the dielectric boundary conditions. 
$k_0=\frac{2\pi}{\lambda_{0}}$ is the free-space wavenumber, and $\lambda_{0}=\frac{c}{f_{c}}$ is the wavelength, with $f_{c}$ and $c$ being the carrier frequency and the speed of light. 

The mathematical expression of radiation coefficient $\rho_{m}$ can be derived based on the CMT, 
which characterizes the EM coupling between a multi-mode waveguide and a single-mode PA by coupled-mode equations (CMEs). 
Specifically, define $A_m(\xi)$ and $B(\xi)$ as the physical-field complex amplitudes of guided mode $m$ and the PA,  
where $\xi\in[0,L]$ denotes the local longitudinal coordinate along the PA within a coupling length $L$. 
From perturbative waveguide theory \cite[Sec.~13-3]{SnyderLove1983}, 
the standard CMEs of the multi-mode PASS are given by 
\begin{align}
    \frac{d A_m(\xi)}{d\xi}
        &= -j\beta_m A_m(\xi) - j\kappa_m B(\xi),
        \quad \forall m\in\mathcal{M}, \label{eq:CME_env_A}\\
    \frac{d B(\xi)}{d\xi}
        &= -j\beta^{\mathrm{PA}} B(\xi)
        - j\sum_{m\in\mathcal{M}}\kappa_m^{\ast} A_m(\xi), \label{eq:CME_env_B}
\end{align}
where $\kappa_m$ is the overlap-integral coupling strength, as derived from standard CMT~\cite{HausOpto,Huang1994CMT}. 
$\beta_m$ and $\beta^{\mathrm{PA}}$ denote the propagation constants of 
guided mode $m$ and the eigenmode of PA, respectively. 
To obtain closed-form solutions of the above CMT system, 
we further demonstrate that it can be efficiently decomposed into $M$ individual two-mode CMT subsystems. 

\begin{figure}[!t]
    \centering
    \includegraphics[width=1\linewidth]{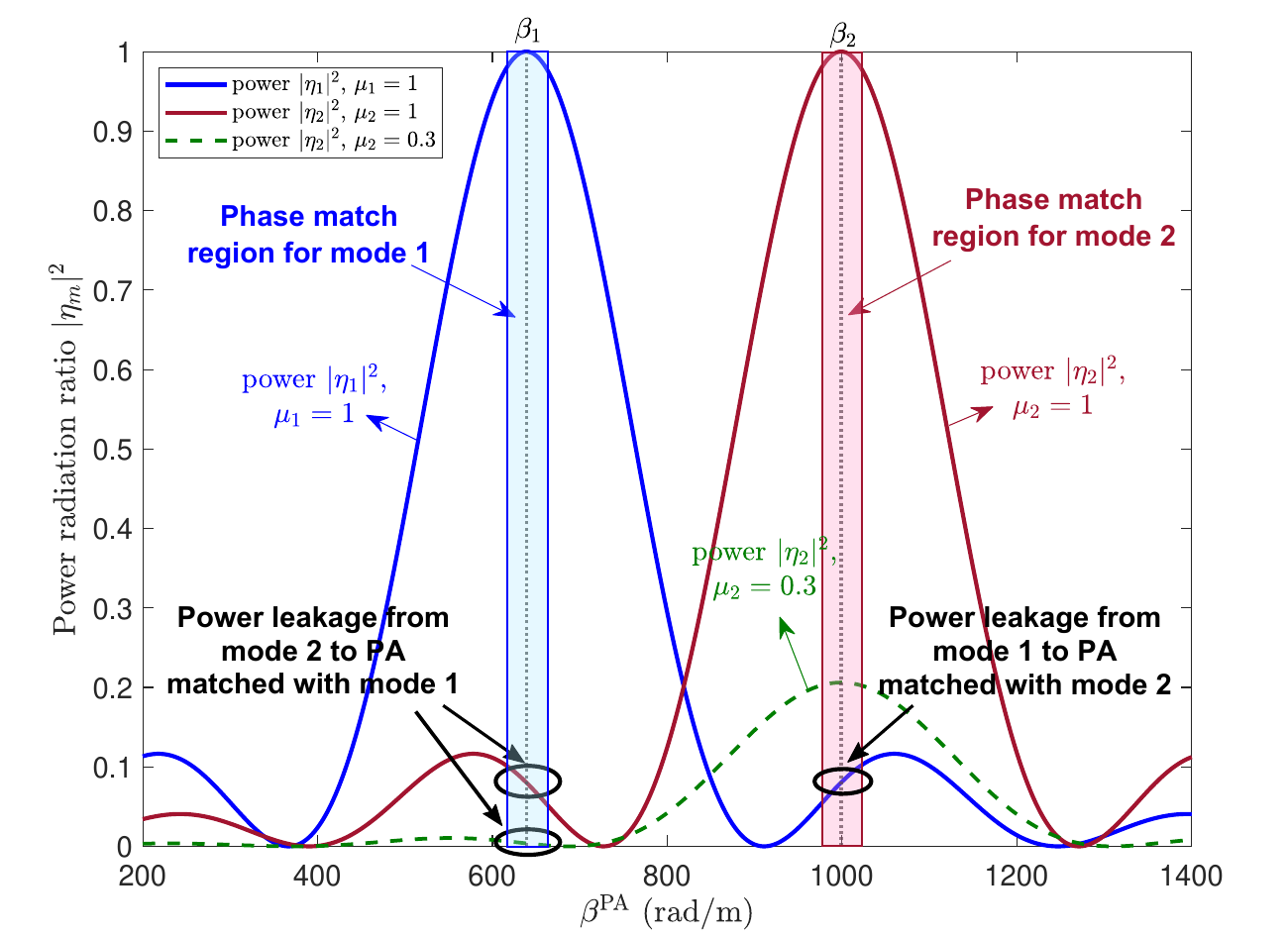}
    \caption{Power radiation ratio $|\eta_m|^2$ for each guided mode $m$ versus the propagation constant $\beta^{\mathrm{PA}}$  of a PA. 
    $\beta_1=638.8$ rad/m, $\beta_2=999.0$ rad/m, $\mu_1=1$, $\mu_{2}\in\{0.3,1\}$, and coupling strength $L=2$ cm. 
    The PA can predominantly match with one guided mode $m$ at $\beta^{\mathrm{PA}}\approx \beta_{m}$ to radiate its signal power, 
    while leaking a small amount of power from the unmatched mode,  
    which indicates the mode selectivity. 
    }\label{fig:phase_selectivity}
\end{figure}

\begin{proposition}[Decomposition into $M$ Two-Mode Subsystems]
\label{prop:MS_two_mode}
Under mode-separation and weak inter-mode transfer assumptions, 
the multi-mode CME system \eqref{eq:CME_env_A}-\eqref{eq:CME_env_B}
approximately decomposes into $M$ independent two-mode subsystems. 
Then, the power radiation coefficient in \eqref{waveguide_to_PA_single} is given by
\begin{equation}\label{eq:MS_g_singlePA}
    \rho_{m}= \eta_{m}\triangleq\frac{\kappa_{m}}{\phi_{m}} \sin(\phi_{m} L).
\end{equation}
where $\phi_m\triangleq \sqrt{|\kappa_{m}|^2 + \Bigl(\frac{\Delta\beta_m}{2}\Bigr)^2}\in\mathbb{R}$ is the generalized coupling strength, $\eta_{m}$ is the coupling coefficient, 
and $\Delta\beta_m \triangleq \beta^{\mathrm{PA}}-\beta_m$ is the phase mismatch between guided mode $m$ and the PA. 
\begin{IEEEproof}
    See Appendix \ref{proof:prop:MS_two_mode}. 
\end{IEEEproof}
\end{proposition}

\begin{remark}\label{remark:phase_matching}
    When phases are matched, i.e., $\Delta\beta_{m}=\beta_m-\beta^{\mathrm{PA}}=0$, 
    $\eta_{m}$ achieves its maximum:
    \begin{equation}
    \eta_{m}=\sin\!\left(|\kappa_{m}|L\right),
    \end{equation}
    by engineering PA to make $\angle\kappa_{m}=0$ at the matched phase.
\end{remark}
    
\textbf{Mode Selectivity in PA Power Radiation:} Eq.~\eqref{eq:MS_g_singlePA} reveals the \emph{mode selectivity} of PA power radiation. 
To further illustrate it, Fig.~\ref{fig:phase_selectivity} provides an example of power radiation ratios 
for two guided modes versus the propagation constant $\beta^{\mathrm{PA}}$ at one PA. 
The coupling strength is modelled as $|\kappa_m|=\mu_m|\kappa_{m,0}|$, 
where $|\kappa_{m,0}|$ is the maximum coupling strength for mode $m$ under an ideal field matching, 
and $\mu_m\in[0,1]$ denotes the field selectivity of the PA (including spatial overlap, polarization mismatch, and structural asymmetry). 
It can be observed that the PA can predominantly match with one guided mode $m$ at $\beta^{\mathrm{PA}}\approx \beta_{m}$ 
to realize high power radiation (i.e., the phase match region in Fig.~\ref{fig:phase_selectivity}), 
which indicates the mode selectivity of PA power radiation. 
In this case, only a small amount of signal power from the unmatched mode can be radiated, which is termed as \textit{power leakage} in multi-mode PASS, as illustrated in Fig.~\ref{fig:phase_selectivity}. 
Moreover, the power leakage effect from an unmatched mode $m$ to the PA can be decreased by reducing the field selectivity $\mu_{m}$. 

\subsection{Physic Model for Multiple PAs} 
We then consider the multi-PA case, where $N$ PAs are deployed along the multi-mode waveguide. 
The signal launched by feed $m$ is propagated via guided mode $m$ and radiated sequentially by PAs $1,2,\dots,N$. 
The signal propagation from feed $m$ to PA $n$ via guided mode $m$ is represented by
\begin{equation}\label{eq:g_derivation}
    g_{n,m} = \rho_{n,m}e^{-j\beta_{m}x_{n}} = \eta_{n,m}\overline{\eta}_{n,m}^{\mathrm{in}}e^{-j\beta_{m}x_{n}},
\end{equation}
where $\eta_{n,m}$ characterizes the EM coupling between PA $n$ and the residual in-waveguide signal of guided mode $m$, 
and $\overline{\eta}_{n,m}^{\mathrm{in}}$ denotes the complex amplitude of incident signal for guided mode $m$ at PA $n$. 
From \textbf{Proposition \ref{prop:MS_two_mode}},  $\eta_{n,m}$ is given by
\begin{equation}\label{eq:MS_g_final}
\eta_{n,m}=\frac{\kappa_{n,m}}{\phi_{n,m}} \sin(\phi_{n,m} L),
\end{equation}
where $\phi_{n,m}\triangleq \sqrt{|\kappa_{n,m}|^2 + \Bigl(\frac{\Delta\beta_{n,m}}{2}\Bigr)^2}\in\mathbb{R}$ is the generalized coupling strength at PA $n$, 
and $\beta_{n,m}\triangleq \beta_{n}^{\mathrm{PA}}-\beta_m$ is the phase mismatch. 
Specifically, $|\eta_{n,m}|^2$ denotes the power radiation fraction relative to the incident power from guided mode $m$ at PA $n$,
while $1-|\eta_{n,m}|^2$ corresponds to the remaining fraction in the waveguide. 
Hence, the incident power ratio at PA $n$ relative to the initially launched modal power $|c_{m}|^2$ is obtained recursively by 
$\Big|\overline{\eta}_{n,m}^{\mathrm{in}}\Big|^2 = \prod\limits_{i<n}\left(1-|\eta_{i,m}|^{2}\right)$, which yields
\begin{equation}\label{power_ratio_recursive}
    \overline{\eta}_{n,m}^{\mathrm{in}} = \prod_{i<n}\sqrt{\left(1-|\eta_{i,m}|^{2}\right)}.
\end{equation}
Given the mode selectivity of PA power radiation, adjusting multi-mode power radiation ratios for multiple PAs are quite complicated, 
which requires sophisticated designs for the multi-mode PASS. 
In the next section, we will propose a practical scheme that achieves average power radiation and tractable implementations. 

\section{Signal Model and Problem Formulation}
\label{sec:SystemModel}

\begin{figure}[!t]
    \centering
    \includegraphics[width=1\linewidth]{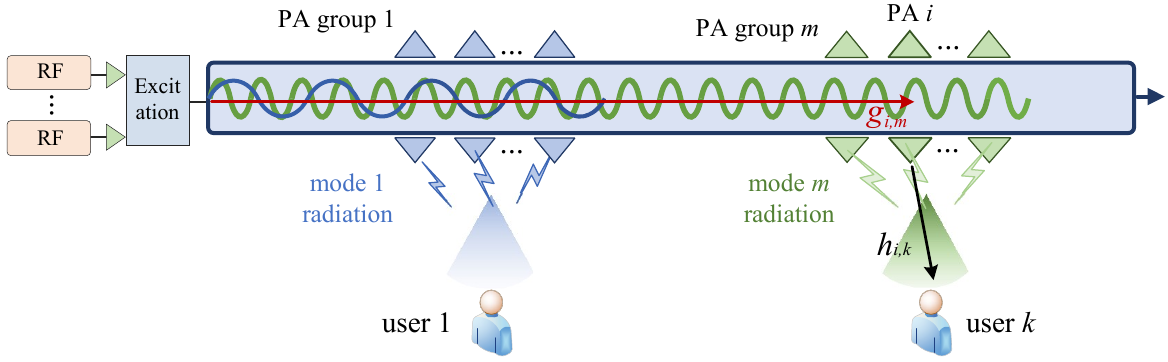}
    \caption{The proposed multi-mode PASS framework.}\label{fig:multimodePASS}
\end{figure}

To investigate the performance of multi-mode PASS, we consider a two-user downlink communication system,  
which simultaneously transmits two independent data streams using only a single waveguide, as shown in Fig. \ref{fig:multimodePASS}. 
Specifically, the signals are carried by $M=2$ guided modes. Each guided mode is excited via an individual feed point, which is termed a feed throughout the paper. 
A set of $N$ PAs, indexed by $\mathcal{N}=\{1,2,\dots,N\}$, are activated along the waveguide, which radiate the multi-mode guided signals into free space. 
The radiated wireless signals are then received by $K=2$ single-antenna mobile users in the downlink.
The sets of guided modes and users are denoted by $\mathcal{M}=\{1,2\}$ and $\mathcal{K}=\{1,2\}$. 
A three-dimension coordinate system is considered. The position of user $k$ is defined as $(a_{k}, y_{k}, 0)$.
The location of PA $n$ is denoted by $(x_{n},0,h_{\mathrm{PA}})$, 
where the coordinate along the $x$-axis is adjustable within $x_{n}\in[0,L_{\mathrm{wg}}]$, and the height $h_{\mathrm{PA}}$ and the waveguide length $L_{\mathrm{wg}}$ are fixed.

\subsection{Signal Model}
The complex data symbols intended for users are denoted by $\mathbf{s} = [s_1, s_2]^{\mathsf T} \in \mathbb{C}^{2\times 1}$.
A digital baseband precoder $\mathbf{W} \in \mathbb{C}^{2\times 2}$ maps data symbol vector $\mathbf{s}$ to the baseband signals that are input to $N_{\mathrm{RF}}$ RF chains. 
Without loss of generality, $N_{\mathrm{RF}}=K=2$ RF chains are adopted to satisfy the minimum requirements of spatial multiplexing. 
At the waveguide entrance, a multi-port launcher excites the target $M$ guided modes (e.g., the TE$_{11}$ and TM$_{01}$ modes) through $M$ feeds.
The resulting modal excitation vector is given by
\begin{align}
    \mathbf c
    = [c_1, c_2]^{\mathsf T}
    = \mathbf{V} \mathbf{W} \mathbf{s} =\mathbf{W} \mathbf{s},
\end{align}
where $c_m$ denotes the complex signal of guided mode $m$. 
We assume that the launcher has been properly calibrated such that the RF-chain outputs are directly mapped to the modal
excitations. Hence, the launcher excitation matrix 
$\mathbf{V}\in\mathbb{C}^{M\times N_{\mathrm{RF}}} = \mathbf{I}_{2\times 2}$ 
reduces to the identity matrix, where each feed can excite one guided mode.

The in-waveguide propagation to PA $n$ via guided mode $m$ is captured by $g_{n,m}$ as:
\begin{equation}
    g_{n,m} = \rho_{n,m}e^{-j\beta_{m}x_{n}}, ~ \forall n \in \mathcal{N}, ~ m \in \mathcal{M},
\end{equation}
where $\rho_{n,m}$ is given by \eqref{eq:g_derivation}. 
Thereafter, PA $n$ further radiates the coupled signals into free space, which are further received by user $k$ through wireless channel. 
The LoS-dominant wireless channel from PA $n$ to user $k$ is given by the geometric free-space spherical-wavefront model \cite{Zhang2022Beamfocusing}:
\begin{equation}\label{eq:h_definition}
    h_{n,k}=\frac{\lambda_{0}}{4\pi}\frac{e^{jk_0 R_{n,k}}}{R_{n,k}},
\end{equation}
where $R_{n,k}=\sqrt{(x_{n}-a_{k})^2+y_{k}^{2}+h_{\mathrm{PA}}^2}$ is the distance between PA $n$ and user $k$. 

Hence, the received signal at user $k$ is given by 
\begin{align}
    y_{k}=\underset{\text{desired signal}}{\underbrace{\sum_{m\in\mathcal{M}}\mathbf{h}_{k}^{\mathsf H}\mathbf{g}_{m}w_{m,k}s_{k}}}
    +\underset{\text{interference}}{\underbrace{\sum_{k'\ne k}\sum_{m\in\mathcal{M}}\mathbf{h}_{k}^{\mathsf H}\mathbf{g}_{m}w_{m,k'}s_{k'}}}+n_{k},
\end{align}
where $\mathbf{g}_m\triangleq[g_{1,m},\dots,g_{N,m}]^{\mathsf T}\in\mathbb{C}^{N\times 1}$ is in-waveguide propagation vector of guided mode $m$, 
$\mathbf{h}_k\triangleq[h_{1,k},\dots,h_{N,k}]^{\mathsf T}\in\mathbb{C}^{N\times 1}$ is wireless channel vector of user $k$, 
and $n_{k}\sim \mathcal{N}(0,\sigma^2)$ is the additive white Gausian noise (AWGN). 
Let $\mathbf{G}=[g_{n,m}]\in\mathbb{C}^{N\times M}$  and $\mathbf{H}\triangleq \left[\mathbf h_1, \mathbf h_2\right] \in \mathbb{C}^{N \times K}$  
denote in-waveguide propagation and wireless channel matrix. 
Moreover, denote $\mathbf y=[y_1,y_2]^{\mathsf T}$ and $\mathbf n=[n_1,n_2]^{\mathsf T}$.
The received signals of users can be compactly written as 
\begin{align}\label{eq:signal}
    \mathbf{y} =
    \mathbf{H}^{\mathsf H}\mathbf{G}\mathbf{W}\mathbf{s} + \mathbf{n}.
\end{align}
Denote $\widetilde{\mathbf{h}}_{k}^{\mathsf H}\triangleq \mathbf{h}_{k}^{\mathsf{H}}\mathbf{G}\in\mathbb{C}^{1\times M}$. 
The signal-to-interference-and-noise ratio (SINR) of user $k$ is given by
\begin{equation}\label{SINR_definition}
    \mathrm{SINR}_{k}= \frac{|\widetilde{\mathbf{h}}_{k}^{\mathsf H}\mathbf{w}_{k}|^2}{|\widetilde{\mathbf{h}}_{k}^{\mathsf H}\mathbf{w}_{k'}|^2+\sigma^{2}}, 
    ~ k' \ne k.
\end{equation}

\subsection{PA Grouping and Power Leakage Regimes}
Owing to the mode selectivity of PA radiation (as demonstrated in Fig. \ref{fig:phase_selectivity}), 
we group multiple PAs to match different modes, thus enabling mode-selective power radiation with limited leakage.
The total number of $N$ PAs are divided into $M=2$ disjoint groups, denoted by $\mathcal{N}_{1}$ and $\mathcal{N}_{2}$, which satisfies
$\mathcal{N} = \mathcal{N}_{1} \cup \mathcal{N}_{2}$. 
PA group $m$ is designed to match guided mode $m$, 
where each PA $n\in\mathcal{N}_{m}$ satisfies $\beta_{n}^{\mathrm{PA}} = \beta_{m}$.
The phase mismatch $\Delta\beta_{n,m}$ between guided mode $m$ 
and the $n$-th PA in group $m'$ becomes
\begin{equation}
    \Delta\beta_{n,m}=\begin{cases}
        0, & m= m',\\
       \Delta\beta_{m,m'}, & m\neq m',
        \end{cases} 
        \qquad \forall n\in\mathcal{N}_{m'},
\end{equation}
where the phase difference $\Delta\beta_{m,m'} \triangleq\beta_{m}-\beta_{m'}$. 

For ease of design, we adopt an average power radiation within each PA group,   
where all PAs in group $m$ share an average power radiation ratio for guided mode $m$:
\begin{equation}\label{equal_power_radiation}
    |\rho_{1,m}| = |\rho_{2,m}| =  \dots = |\rho_{|I_{m}|,m}| = \frac{\rho_m^{\mathrm{group}}}{\sqrt{|\mathcal{N}_{m}|}}.
\end{equation}
$\rho_m^{\mathrm{group}}=\prod_{n'\notin\mathcal{N}_{m}}\sqrt{1-|\eta_{n',m}|^{2}}$ 
is the power ratio allocated to group $m$ by subtracting the leaked power of mode $m$.

\subsubsection{Non-Leakage Regime}
In this regime, each guided mode $m$ only transfers energy with PA group $m$, but not leaks power to unmatched PAs. 
The coupling coefficient $\eta_{n,m}$ is given by
\begin{equation}\label{coupling_non_leakage}
    \eta_{n,m}^{\mathrm{NL}}=\begin{cases}
        \sin(|\kappa_{n,m}|L), & n\in\mathcal{N}_{m},\\
        0, & n\notin\mathcal{N}_{m},
        \end{cases}
\end{equation} 
where coupling strength $|\kappa_{n,m}|$ is determined by the spacing $S_{n}$ between PA $n$ and the waveguide \cite[Proposition 1]{Xu2025PASSPower}: 
\begin{equation}\label{eq:kappa_match}
    |\kappa_{n,m}| =  e^{-\alpha_{m}S_{n}} |\Omega_{m}|,  \quad \forall n \in \mathcal{N}_{m}.
\end{equation}
The non-leakage regime leads to $\rho_{m}^{\mathrm{group}}=1$ since $\eta_{n,m}^{\mathrm{NL}}=0$ for any $n \notin \mathcal{N}_{m}$.
From the definition of \eqref{eq:g_derivation}, achieving average power radiation for group $m$ in \eqref{equal_power_radiation} requires that
\begin{equation} 
    \eta_{n,m}^{\mathrm{NL}} = \frac{1}{\sqrt{|\mathcal{N}_{m}|-n+1}}, ~\forall n \in \mathcal{N}_{m}.
\end{equation}
Therefore, the spacing between PA $n$ and the waveguide needs to be configured by
\begin{equation}\label{spacing_PA}
    S_{n}=-\frac{1}{\alpha_{m}}\ln\left(\frac{\arcsin\left(\eta_{n,m}^{\mathrm{NL}}\right)}{L|\Omega_{m}|}\right), 
    ~ \forall n \in \mathcal{N}_{m}.
\end{equation}
As a result, 
the effective channel coefficient is reduced to
\begin{equation}
    \widetilde{h}_{m,k}^{\mathrm{NL}}= \mathbf{h}_{k}^{\mathsf{H}}\mathbf{g}_{m}^{\mathrm{NL}} = 
    \sum_{n\in\mathcal{N}_{m}}\frac{\lambda_{0}}
    {4\pi R_{n,k}\sqrt{\left|\mathcal{N}_{m}\right|}}e^{-j(\frac{2\pi}{\lambda_{0}}R_{n,k}+\beta_{m}x_{n})}.
\end{equation}

\subsubsection{Weak-Leakage Regime}
For the weak-leakage regime, the coupling coefficient $\eta_{n,m}^{\mathrm{WL}}$ is given by
\begin{equation}\label{coupling_weak_leakage}
    \eta_{n,m}^{\mathrm{WL}}=\begin{cases}
        \sin(|\kappa_{n,m}|L), & n\in\mathcal{N}_{m},\\
        \frac{\sin\left(\sqrt{\left|\kappa_{n,m}\right|^{2}+\tfrac{\left|\Delta\beta_{n,m}\right|^{2}}{4}}L\right)\kappa_{n,m}}{\sqrt{\left|\kappa_{n,m}\right|^{2}+\tfrac{\left|\Delta\beta_{n,m}\right|^{2}}{4}}}, & n\notin\mathcal{N}_{m}.
        \end{cases}
\end{equation}
As in the non-leakage regime, we enforce average power radiation within each PA group. 
Specifically, for the PAs assigned to mode $m$, we set
\begin{equation}\label{eta_equal_weak_leakage}
    \eta_{n,m}^{\mathrm{WL}} = \frac{\rho_m^{\mathrm{group}}}{\sqrt{|\mathcal{N}_{m}|-n+1}}, \quad \forall n \in \mathcal{N}_{m}.
\end{equation}
which can be satisfied by choosing the PA-waveguide spacing $S_{n}$ similar to \eqref{spacing_PA}.  

To mitigate leakage, multiple PA groups are placed sequentially along the waveguide. 
With this arrangement, the power of guided mode $m$ is radiated by PA groups in a predetermined ordered, 
and thus can only leak to groups $m'<m$ that radiate earlier than the intended group $m$. 
This ordering also enables a tractable group-by-group design, i.e., once the spacing $\{S_{n}\}$ of a group are determined by \eqref{spacing_PA}, the corresponding coupling 
strength between PA $n$ and any unmatched mode $m$ (i.e., $n\notin \mathcal{N}_{m}$) are determined as 
$|\kappa_{n,m}| =  e^{-\alpha_{m}S_{n}} |\Omega_{m}|\mu_{m}^{\mathrm{unmatch}}$. 
Using \eqref{coupling_weak_leakage} and the fact that $\sin(\cdot)\leqslant 1$, 
the power leakage from mode $m$ to any unmatched PA $n \notin \mathcal{N}_{m}$ is upper bounded by 
\begin{equation}\label{upper_bound_power_leakage}
    |\eta_{n,m}^{\mathrm{WL}}|^2 \leqslant\frac{|\kappa_{n,m}|^2}{1+\Delta\beta_{m,m'}^{2}/4}, ~\forall n \notin \mathcal{N}_{m}.
\end{equation} 
Therefore, the leakage is inherently weak when the phase mismatch $|\Delta \beta_{m,m'}|$ is sufficiently large and 
the mismatch-induced field selectivity $\mu_{m}^{\mathrm{unmatch}}$ is small. 

As a result, the effective channel coefficient becomes 
\begin{equation}
    \begin{split}
        \widetilde{h}_{m,k}^{\mathrm{WL}}&=\mathbf{h}_{k}^{\mathsf{H}}\mathbf{g}_{m}^{\mathrm{WL}}
        =\sum_{n\in\mathcal{N}_{m}}\frac{\lambda_{0}\rho_{m}^{\mathrm{group}}}{4\pi R_{n,k}\sqrt{\left|\mathcal{N}_{m}\right|}}e^{-j(\frac{2\pi}{\lambda_{0}}R_{n,k}+\beta_{m}x_{n})}
        \\&+\sum_{n\in\mathcal{N}_{m'},~m'<m}\frac{\lambda_{0}}{4\pi R_{n,k}}\rho_{n,m}e^{-j(\frac{2\pi}{\lambda_{0}}R_{n,k}+\beta_{m}x_{n})}.
    \end{split}
\end{equation}

\subsection{Problem Formulation}
To maximize the system sum rate, 
we formulate the following joint PA position and baseband beamforming optimization problem 
subject to each user's minimum data rate requirement:
\begin{subequations}
    \label{eq:P0}
    \begin{align}
    \max_{\mathbf{W},\mathbf{x}}& ~\sum_{k\in\mathcal{K}} \log_2\!\big(1+\mathrm{SINR}_{k}\big)
    \label{eq:ms_bilevel_upper_obj}\\
    \mathrm{s.t.}\quad&
    \sum_{k} \|\mathbf{w}_{k}\|^2 \leqslant P_{\max}, \label{eq:P0_power} 
    \\& \mathrm{SINR}_{k} \geq \mathrm{SINR}_{k}^{\min}, ~\forall k\in\mathcal{K}, \label{eq:sinr_min_constraint}
    \\&x_{n+1}-x_{n} \geq \Delta_{\min},~ i=1,\ldots,I-1, \label{eq:P0_x_cons1} 
    \\& x_{\min}\le x_1 < \cdots < x_I\le x_{\max},\label{eq:P0_x_cons2} 
    \end{align}
\end{subequations}
where $\mathrm{SINR}_{k}$ is computed by $\widetilde{h}_{m,k}^{\mathrm{NL}}$ in non-leakage regime and $\widetilde{h}_{m,k}^{\mathrm{WL}}$ in weak-leakage regime, 
constraint \eqref{eq:P0_power} guarantees the maximum transmit power $P_{\max}$, 
constraint \eqref{eq:sinr_min_constraint} ensures the minimum rate requirement of each user $k$, 
constraint \eqref{eq:P0_x_cons1} indicates the minimum inter-PA spacing, 
and constraint \eqref{eq:P0_x_cons2} denotes the sequential ordering and motion ranges of PAs.

Note that the term $e^{-j\beta_m x_{n}}$ is a complex exponential function of PA position $x_{n}$. 
Hence, even a small wavelength-scale perturbation of the PA location may induce a significant phase rotation,  
making the equivalent channel highly oscillatory with respect to PA positions.
Moreover, the couped position and baseband beamforming need to be optimized in a joint and coordinated manner. 
Hence, problem \eqref{eq:P0} is a high-oscillatory, strong-coupled, and non-convex problem. 
In the sequel, we will first investigate the non-lekage two-PA case, and then consider a general multi-PA case.

\section{Solution for Non-Leakage, Two-PA Case}
To provide insights into the system design, in this section we consider  a two-PA, non-leakage case, and derive a channel orthogonality based solution to solve \eqref{eq:P0}. 
In this case, each PA group contains only one PA, i.e., $|\mathcal{N}_{1} |=|\mathcal{N}_{2}| = 1$. 
Hence, the $m$-th PA is directly matched to mode $m$ and radiates its signal at location $x_m$ along propagation direction. 
Moreover, the distance from PA $m$ to user $k$ is denoted by $R_{m,k}$. 
By ignoring power leakage, effective channel coefficient becomes
\begin{equation}
    \widetilde{h}_{m,k}=\frac{\lambda_{0} e^{-j\left(\frac{2\pi}{\lambda_{0}}R_{m,k}+\beta_{m}x_{m}\right)}}{4\pi R_{m,k}}.
\end{equation}


\subsection{Channel Orthogonality Conditions}
\subsubsection{Two-Scale Conditions}
To achieve efficient mode-domain multiplexing for simultaneous two-user transmissions, 
the proposed solution ensures the following channel orthogonality conditions, and thus realize interference nulling:
\begin{equation}\label{eq:P1_zero_intf}
\left|\mathbf{\widetilde{h}}_{1}^{H}\mathbf{\widetilde{h}}_{2}\right|^{2}=\left|\mathbf{\widetilde{h}}_{2}^{H}\mathbf{\widetilde{h}}_{1}\right|^{2}=0.
\end{equation}
In \eqref{eq:P1_zero_intf}, $\left|\mathbf{h}_{k'}^{H}\mathbf{h}_{k}\right|^2$ can be rearranged as 
\begin{equation}
        \left|\mathbf{h}_{k'}^{H}\mathbf{h}_{k}\right|^2
        =\frac{\lambda_{0}^{2}}{16\pi^{2}}\left|\sum_{m\in\mathcal{M}}\frac{1}{R_{m,1}R_{m,2}}e^{j\frac{2\pi}{\lambda_{0}}\left(R_{m,k'}-R_{m,k}\right)}\right|^2.
\end{equation}

Denote $b_{m}\triangleq\tfrac{1}{R_{m,1}R_{m,2}}$ and $\theta_{m}\triangleq\tfrac{2\pi}{\lambda_{0}}(R_{m,2}-R_{m,1})$. 
From triangle inequality, we have 
\begin{equation}
    \begin{split}
    \left|b_{1}e^{j\theta_{1}}+b_{2}e^{j\theta_{2}}\right|^{2}
    &=\left|b_{1}\right|^{2}+\left|b_2\right|^{2}+2\left|b_{1}\right|\left|b_{2}\right|\cos\left(\theta_{1}-\theta_{2}\right)\\
    &\geq\left(\left|b_{1}\right|-\left|b_{2}\right|\right)^{2},
\end{split}
\end{equation}
where the equality is achieved when $\cos\left(\theta_{1}-\theta_{2}\right)=-1$, i.e., $\theta_{1} - \theta_{2} = (2n+1)\pi$, $n\in\mathbb{Z}$, 
implying that the two-mode coherent waves with phase differences $\theta_{1}-\theta_{2}$ achieves destructive interference. 
Therefore, the interference nulling constraint \eqref{eq:P1_zero_intf} can be enforced by a magnitude equality constraint 
$|b_{1}| = |b_{2}|$ and a phase difference constraint $\theta_{1} - \theta_{2} = (2n+1)\pi$, which are equivalent to
\begin{equation}\label{P1_intf_nullling_largescale}
R_{1,1}R_{1,2} = R_{2,1}R_{2,2},
\end{equation}
\begin{equation}\label{P1_intf_nullling_smallscale}
\left(R_{1,2}-R_{1,1}\right) - \left(R_{2,2}-R_{2,1}\right) = \left(n+\tfrac{1}{2}\right)\lambda_{0}, ~n \in \mathbb{Z}.
\end{equation}
We can observe that constraints \eqref{P1_intf_nullling_largescale} and \eqref{P1_intf_nullling_smallscale} alter $R_{m,k}$ at two different scales.
The interference magnitude is eliminated in \eqref{P1_intf_nullling_largescale} by tuning $R_{m,k}$ at a large distance scale, 
whereas the destructive interference for two-mode coherent waves is ensured by the phase difference constraint \eqref{P1_intf_nullling_smallscale} by tuning $R_{m,k}$ 
at a wavelength scale. 

\subsubsection{Problem Reduction}
Given any PA positions $\mathbf{x}$ that ensures the channel orthogonality, the optimal beamforming solution is given by the MRT beamforming strategy: 
\begin{equation}\label{eq:MRT_BF}
    \mathbf{w}_{k}^{\ast}=\frac{\sqrt{p_{k}}}{\left\Vert \mathbf{\widetilde{h}}_{k}\right\Vert }\mathbf{\widetilde{h}}_{k},
\end{equation}
and the optimal power allocation $p_{k} = \max\Big\{\mu-\frac{\sigma^2}{\|\mathbf{\widetilde{h}}_{k}\|^2},0\Big\}$ is given by water-filling strategy 
with $\mu$ being the water-filling level. 
Combining \eqref{eq:MRT_BF} and \eqref{eq:P1_zero_intf}, the SINR of user $k$ becomes
\begin{equation}\label{SINR_distance}
        \mathrm{SINR}_{k}
        =\frac{p_{k}\left\Vert \mathbf{\widetilde{h}}_{k}\right\Vert^2 }
        {\frac{p_{k'}}{\left\Vert \mathbf{\widetilde{h}}_{k'}\right\Vert ^{2}}|\widetilde{\mathbf{h}}_{k}^{H}\widetilde{\mathbf{h}}_{k'}|^{2}+\sigma^{2}}
        =\frac{p_{k}}{\sigma^{2}}\frac{\lambda_{0}^{2}}{16\pi^{2}}
        \sum_{m\in\mathcal{M}}\frac{1}{R_{m,k}^{2}},
\end{equation}
which also depends on the large-scale adjustment of $R_{m,k}$. 
Hence, the minimum SINR constraints can be recast as  
\begin{equation}\label{P1_distance_max}
    \sum_{m\in\mathcal{M}}\frac{1}{R_{m,k}^{2}} \geq \frac{16\pi^2\sigma^2}{p_{k}\lambda_{0}^2} \mathrm{SINR}_{k}^{\min}, ~ \forall k \in \mathcal{K}.
\end{equation}
Hence, the sum rate maximization problem can be rearranged as the following sum reciprocal distance minimization under the dual-scale equality constraints:
\begin{subequations}\label{eq:P1_twoscale}
    \begin{align}
    \max_{\mathbf{x}}~ &
    \sum_{k\in\mathcal{K}}\log_{2}\left(1+\frac{p_{k}}{\sigma^{2}}\frac{\lambda_{0}^{2}}{16\pi^{2}}
    \sum_{m\in\mathcal{M}}\frac{1}{R_{m,k}^{2}}\right)
    \label{eq:P1_twoscale_obj}\\
    \mathrm{s.t.}\quad &
    \eqref{eq:P0_x_cons1}, \eqref{eq:P0_x_cons2}, \eqref{P1_intf_nullling_largescale}, \eqref{P1_intf_nullling_smallscale}, \eqref{P1_distance_max}. 
    \end{align}
\end{subequations}
At first glance, problem \eqref{eq:P1_twoscale} is nonconvex over $\mathbf{x}$. 
The key observation is that the two equality constraints \eqref{P1_intf_nullling_largescale}-\eqref{P1_intf_nullling_smallscale} 
severely restrict the feasible set, 
allowing the original two-dimension continuous optimization problem to be reduced to a one-dimension search over a finite set of candidate points.

\subsection{PA Position Optimization}
In this subsection, we exploit the intrinsic geometric structure of the considered case 
to efficiently obtains the global optimum without resorting to exhaustive search.

\subsubsection{Structural Analysis}
\label{subsec:structural_opt}
As shown below, the equality constraint 
\eqref{P1_intf_nullling_largescale} enables an exact dimension reduction 
from a two-dimension continuous problem to a one-dimension search problem over a finite set of candidate points. 
Denote $z_{k} \triangleq \sqrt{y_{k}^2 + h_{\mathrm{PA}}^2}$. 
Moreover, define the auxiliary function
\begin{equation}\label{P1_product_constraint} 
D(x) \!\triangleq\! R_1(x)R_2(x)
\!=\! \sqrt{(x-a_1)^2+z_1^2}\,
  \sqrt{(x-a_2)^2+z_2^2}.
\end{equation}
Constraint \eqref{P1_intf_nullling_smallscale} can be equivalently written as
\begin{equation}
D(x_1)=D(x_2).
\end{equation}

In practice, the optimal PA positions can be searched within region $x_{n} \in [a_{1}, a_{2}]$,  
thus ensuring PAs are closer to both users to minimize the path loss in the objective \eqref{eq:P1_twoscale_obj}. 
For ease of notations, define $p\triangleq\frac{-d^2+2(z_1^2+z_2^2)}{4}$, $q\triangleq\frac{d(z_2^2-z_1^2)}{4}$, 
and cube discriminant $\Delta_{D} \triangleq \Big(\frac{q}{2}\Big)^2+\Big(\frac{p}{3}\Big)^3$.
The following corollary characterizes the structure of $D(x)$ over $x \in [a_{1}, a_{2}]$.
    
\begin{corollary}[Stationary points of $D(x)$]\label{corolllary:stationary_D}
    Function $D(x)$ has $S\leqslant 3$ stationary points on region $[a_1,a_2]$. 
    If $\Delta_{D}\ge0$, $D(x)$ has a unique stationary point
    \begin{equation}
        \tau_1=a_1+\frac{a_2-a_1}{2}+\operatorname{cbrt}\Big(-\frac q2+\sqrt{\Delta_{D}}\Big)+\operatorname{cbrt}\Big(-\frac q2-\sqrt{\Delta_{D}}\Big),
    \end{equation}
    where $\operatorname{cbrt}(\cdot)$ denotes the real cube root.
    If $\Delta_{D} <0$, $D(x)$ has a number of $S=3$ stationary points, denoted by
    \begin{equation}
        \tau_s=a_1+\frac{a_2-a_1}{2}+2\sqrt{-\frac p3}\cos\!\Big(\frac{\varphi-2\pi(s-1)}{3}\Big), ~ s=1,2,3,
    \end{equation}
    where $\varphi\triangleq\arccos\Big(\frac{3q}{2p}\sqrt{-\frac3p}\Big)$.
\end{corollary}
\begin{proof}
See Appendix \ref{proof:corolllary:stationary_D}.
\end{proof}
Since $D'(a_1)<0$ and $D'(a_2)>0$, $a_1$ and $a_2$ are not stationary points. 
The stationary points of $D(x)$ in $[a_1,a_2]$ can be ordered as 
$a_1<\tau_1<\cdots<\tau_S<a_2$.
Define $\tau_0\triangleq a_1$ and $\tau_{S+1}\triangleq a_2$. 
Then, $[a_1,a_2]$ consists of $S+1$ intervals $\mathcal{X}_{s} = [\tau_s,\tau_{s+1}]$, $s=0,1,\dots,S$, with the following property.
\begin{proposition}[Piecewise monotonicity of $D(x)$]\label{prop:D_piecewise_monotone}
    Function $D(x)$ is strictly monotone in each interval $\mathcal{X}_{s}= [\tau_s,\tau_{s+1}]$ partitioned by stationary points, $\forall s=0,1,\dots,S$. 
\end{proposition}
The piecewise monotonicity of $D(x)$ in \textbf{Proposition~\ref{prop:D_piecewise_monotone}} 
ensures that given $x_1\in\mathcal{X}_s$,  in each interval $\mathcal{X}_{s'}$ at most one $x_2$ satisfies $D(x_1) = D(x_2)$ (and $s'>s$ since $x_2>x_1$). 
Because $S\leqslant 3$, for each fixed $x_1$, at most $3$ feasible points of $x_2$ exist. 
Then, the inequality constraints \eqref{eq:P0_x_cons1} and \eqref{P1_distance_max} 
can be checked directly, and any candidate violating these inequality constraints can be discarded. 
Hence, constraint \eqref{P1_product_constraint} induces a single-valued mapping within each interval:
\begin{equation}
x_2=\mathcal{Z}(x_1), \qquad x_1\in\mathcal{X}_{s},~ x_2\in\mathcal{X}_{s'}, ~s'>s.
\end{equation}
Thus, all feasible solutions satisfying \eqref{P1_intf_nullling_largescale} lie on a one-dimension curve parameterized by $x_1$. 
For each $x_1\in\mathcal{X}$, the corresponding $x_2=\mathcal{Z}(x_1)$ can be coarsely determined to minimize large-scale path losses 
and satisfy the large-scale equality constraint \eqref{P1_intf_nullling_largescale}. 
To additionally enforce phase-difference constraint \eqref{P1_intf_nullling_smallscale}, we further invoke Newton's method to refine the  
estimated $(x_1,\mathcal{Z}(x_1))$ in the wavelength scale.  

\subsubsection{Newton-based One-Dimension Search Algorithm}
Define the reduced one-dimension objective
\begin{equation}
F(x_1)=
\sum_{k\in\mathcal{K}}
\log_2\!\left(
1+\frac{p_{k}}{\sigma^{2}}\frac{\lambda_{0}^{2}}{16\pi^{2}}\Big(\frac{1}{R_k^2(x_1)}+\frac{1}{R_k^2(\mathcal{Z}(x_1))}\Big)
\right).
\end{equation}
Since all feasible points of the original problem are contained in the finite set $\mathcal{X}$, the global optimum is obtained by selecting
\begin{equation}
x_1^\star=\mathop{\arg\max}_{x_1\in\mathcal{X}} F(x_1),\qquad
x_2^\star=\mathcal{Z}(x_1^\star).
\end{equation}

\begin{algorithm}[t]
    \caption{Global PA Position Optimization via Coarse Sampling and Local Refinement}
    \label{alg:global_opt_refined}
    \begin{algorithmic}[1]
    \Require Number of grid points $n_{\mathrm{grid}}$, tolerance $\epsilon$.
    \State Obtain $S$ stationary points $\{\tau_s\}$ 
    using Corollary \ref{corolllary:stationary_D} ($S\leqslant 3$).
    \State Initialize a uniform grid in $\mathcal{X} = \mathcal{X}_{0} \cup  \mathcal{X}_{1}\dots \cup \mathcal{X}_{S-1}$. 
    \State Set $F_{\max} \gets -\infty$.
    \State \textit{// One-dimension coarse sampling}
    \For{each $x_1$ in $\mathcal{X}_s$, $s=0,\dots,S-1$}
        \For{each $s' = s+1, \dots, S$}
        \State Obtain $x_2$ by solving $D(x_2) = D(x_1)$ over interval $x_2\in\mathcal{X}_{s'}$ using bisection method.
        \State If {$x_2 - x_1 < \Delta$}, \textbf{continue;}
        \State\label{alg_refinement_start} \textit{// Local refinement by Newton's method}
        \State $x_1^{(0)} \leftarrow x_1$, $x_2^{(0)} \leftarrow x_2$, $t\leftarrow1$. Estimate $n$ by \eqref{Newton_iteration_integer}. 
        \Repeat
        \State Update $x_{1}^{(t)}$ and $x_{2}^{(t)}$ by \eqref{eq:newton_update}. $t\leftarrow t+1$.
        \Until{reach termination criterion}
        \State If refined points $(x_1', x_2')$ violate $\|\mathbf{F}(x_1', x_2')\|_2\leqslant \epsilon$ or \eqref{P1_distance_max}, \textbf{continue};
        \If{$(x_1', x_2')$ is feasible and $F(x_1', x_2') > F_{\max}$}
            \State $(x_1^\star, x_2^\star) \gets (x_1', x_2')$, 
            $F_{\max} \gets F(x_1', x_2')$.
        \EndIf\label{alg_refinement_end}
    \EndFor
    \EndFor
    \State \textbf{Return} optimal $(x_1^\star, x_2^\star)$ and objective $F_{\max}$.
    \end{algorithmic}
\end{algorithm}

We develop the following global optimization algorithm, which mainly contains the following four steps. 
\begin{enumerate}
\item[\textbf{(i)}] Obtain $\tau_{1},\tau_{2},\dots,\tau_{S}$ ($S\leqslant 3$) by solving $D'(x)=0$.
\item[\textbf{(ii)}] \textbf{One-dimension search:} 
Uniformly sample $x_{1}$ over $[a_1, \tau_{S}]$. 
For each coarse sampling of $x_{1}$ in $\mathcal{X}_{s}$, obtain $x_{2}\in\mathcal{X}_{s'}$, $\forall s' >s$, 
that satisfies large-scale channel-orthogonality constraint \eqref{P1_intf_nullling_largescale} via bisection method.
\item[\textbf{(iii)}] \textbf{Local refinement:} Refine each coarse candidate $(x_1,x_2)$ within small wavelength-scale
$\delta$-neighborhoods $[x_{n}-\delta, x_{n}+\delta]$ 
to enforce both \eqref{P1_intf_nullling_largescale} and
\eqref{P1_intf_nullling_smallscale}. 
These two coupled nonlinear equations can be solved using a two-variable Newton method. 
Specifically, define functions $\mathbf{F}(x_1,x_2)=[f_1(x_1,x_2),\,f_2(x_1,x_2)]^\top$ with
\begin{align}
    f_{1}(x_{1}, x_{2}) &\triangleq D(x_{1}) - D(x_{2}) = 0, \label{eq:local_f1}\\
    f_{2}(x_{1}, x_{2}) &\triangleq \varphi(x_{1}, x_{2}) - \left(n + \tfrac{1}{2} \right)\lambda_{0} = 0, \label{eq:local_f2}
\end{align}
where $\varphi(x_{1}, x_{2}) = [R_{1}(x_{2}) - R_{1}(x_{1})] - [R_{2}(x_{2}) - R_{2}(x_{1})]$ denotes 
the differential phase differences between the two users. 
To preserve global optimality, in \eqref{eq:local_f2} we choose the integer $n$ that minimizes the required adjustment from the coarse point:
\begin{equation}\label{Newton_iteration_integer}
n = \mathrm{round}\!\left(\frac{\varphi(x_1^0,x_2^0)}{\lambda_{0}}-\frac{1}{2}\right). 
\end{equation}
The target phase $(n+\frac{1}{2}\lambda)$ is the closest feasible value to the coarse estimate $\varphi(x_1^0,x_2^0)$. 
Hence, Newton's method only induces wavelength-scale refinement around $(x_1^0,x_2^0)$ to eliminate residual misaligned phase, 
which thus preserve the large-scale path-loss optimality.

Based on Newton's method, we update $(x_1,x_2)$ by solving $\mathbf{F}(x_1,x_2)=\mathbf{0}$ via
\begin{equation}
    \begin{bmatrix}
        x_{1}^{(t+1)} \\
        x_{2}^{(t+1)}
    \end{bmatrix}
    = 
    \begin{bmatrix}
        x_{1}^{(t)} \\
        x_{2}^{(t)}
    \end{bmatrix}
    - 
    \mathbf{J}^{-1}(x_{1}^{(t)}, x_{2}^{(t)})
    \cdot
    \mathbf{F}(x_{1}^{(t)}, x_{2}^{(t)}), \label{eq:newton_update}
\end{equation}
where $\mathbf{J}(x_{1},x_{2})$ is the Jacobian matrix given by
$\mathbf{J} = 
    \begin{bmatrix}
        \tfrac{\partial f_1}{\partial x_1} & \tfrac{\partial f_1}{\partial x_2} \\
        \tfrac{\partial f_2}{\partial x_1} & \tfrac{\partial f_2}{\partial x_2}
    \end{bmatrix}$,
with partial derivatives computed by
\begin{align}
    \frac{\partial f_1}{\partial x_1} &= R_2(x_1) \frac{x_1 - a_1}{R_1(x_1)} + R_1(x_1) \frac{x_1 - a_2}{R_2(x_1)}, \label{eq:df1dx1} \\
    \frac{\partial f_1}{\partial x_2} & 
    = - R_2(x_2) \frac{x_2 - a_1}{R_1(x_2)} - R_1(x_2) \frac{x_2 - a}{R_2(x_2)}, \\
    \frac{\partial f_2}{\partial x_1} &= -\frac{x_1 - a_1}{R_1(x_1)} + \frac{x_1 - a_2}{R_2(x_1)}, \\
    \frac{\partial f_2}{\partial x_2} &= \frac{x_2 - a_1}{R_1(x_2)} - \frac{x_2 - a_2}{R_2(x_2)}.
\end{align}
The iteration in \eqref{eq:newton_update} is initialized from the coarse estimate 
$(x_1^{(0)}, x_2^{(0)})$ generated in step (ii). 
Since the coarse solution is already close to feasibility and
$\mathbf{J}(x_1,x_2)$ is nonsingular for almost all $(x_1,x_2)$ in the
local neighborhood, the method typically converges within a few steps. 
A refined point $(x_1', x_2')$ is accepted if the residual norm $\|\mathbf{F}(x_1', x_2')\|_2$ 
is below the tolerance $\epsilon$ and inequality constraints in \eqref{eq:P1_twoscale} are also satisfied.
\item[\textbf{(iv)}] From all feasible candidate points $(x_1',\mathcal{Z}(x_1'))$, 
evaluate $F(x_1')$ to find the global optimum $(x_1^\star,x_2^\star)$.
\end{enumerate}
The entire procedure is summarized in \textbf{Algorithm \ref{alg:global_opt_refined}}.

\section{Low-Complexity Solution for the General Multiple-PA Case}
\label{sec:de_pso_zf_joint_opt}
In this section, we develop low-complexity solution for the geneal multi-PA case, which may operate in both non-leakage and weak-leakage regimes. 
To tackle \eqref{eq:P0} efficiently we develop a low-complexity PSO-ZF algorithm with a bi-level structure. 
The outer loop updates $\mathbf{x}$ via PSO \cite{Kennedy1995PSO} under the physical constraints, 
while the inner loop computes a closed-form ZF precoder for each candidate configuration. 
The minimum-SINR constraints are enforced in the fitness evaluation by discarding infeasible solutions, 
enabling efficient and robust search despite wavelength-scale oscillations.

\subsection{Baseband Beamforming via Zero Forcing}

For a given PA position vector $\mathbf{x}$, the equivalent baseband channel from the digital
precoder to the $K$ users is given by
\begin{equation}\label{eq:eff_H}
\mathbf{H}_{\mathrm{eff}}(\mathbf{x})
\triangleq
\mathbf{H}^{\mathsf H}(\mathbf{x})\,\mathbf{G}(\mathbf{x})
\in \mathbb{C}^{K \times K},
\end{equation}
where $\mathbf{H}(\mathbf{x})$ denotes the free-space channel between the PAs and the users,
and $\mathbf{G}(\mathbf{x})$ captures the in-waveguide propagation and mode-selective radiation. 
Given a fixed $\mathbf{x}$, the baseband beamformer is designed using zero-forcing precoding on
$\mathbf{H}_{\mathrm{eff}}(\mathbf{x})$.
The unnormalized ZF beamformer is given by
\begin{equation}
\mathbf{W}_{\mathrm{ZF}}(\mathbf{x})
=
\mathbf{H}_{\mathrm{eff}}^{\mathsf H}(\mathbf{x})
\big(
\mathbf{H}_{\mathrm{eff}}(\mathbf{x})\mathbf{H}_{\mathrm{eff}}^{\mathsf H}(\mathbf{x})
\big)^{-1},
\label{eq:zf_precoder}
\end{equation}
where a small diagonal loading may be added for numerical stability.
The beamformer is then scaled to satisfy the total transmit power constraint as
\begin{equation}
\mathbf{W}(\mathbf{x})
=
\mathbf{W}_{\mathrm{ZF}}(\mathbf{x})
\sqrt{\frac{P_{\max}}{\|\mathbf{W}_{\mathrm{ZF}}(\mathbf{x})\|_F^2}}.
\label{eq:zf_power_norm}
\end{equation}
The ZF precoder provides a closed-form and computationally efficient solution that suppresses
multiuser interference in the equivalent baseband domain.

\subsection{Outer-Loop PA Position Optimization}

We next describe the outer-loop PA position optimization based on PSO.
Let $\mathbf{x}=[x_1,\ldots,x_N]^{\mathsf T}$ denote the positions of the $N$ PAs along the waveguide.
The feasible set of PA positions is given by
\begin{equation}
\mathcal{X}
\triangleq
\Big\{
\mathbf{x}:\ 
x_{\min}\le x_1 \le \cdots \le x_N \le x_{\max},\ 
x_{n}-x_{n-1}\ge d_{\min}
\Big\},
\end{equation}
according to constraints \eqref{eq:P0_x_cons1}-\eqref{eq:P0_x_cons2}.

To avoid excessive relocations of PAs and to focus on meaningful refinements around a feasible
initial configuration $\mathbf{x}^{(0)}$, the search is further restricted to a local neighborhood
\begin{equation}
\mathcal{X}_{\delta}(\mathbf{x}^{(0)})
\triangleq
\Big\{
\mathbf{x}\in\mathcal{X}:\ 
|x_{n}-x_{n}^{(0)}|\le \delta,\ \forall n
\Big\},
\end{equation}
where $\delta$ denotes the search radius.
This trust-region constraint reflects practical mechanical limitations on PA relocations and ensures
smooth PA position tuning.

\subsubsection{PSO fitness function}
The PSO algorithm maintains a population of $n_{\rm p}$ particles.
At the $t$-th outer iteration, particle $p$ is characterized by a PA position vector
$\mathbf{x}_p^{(t)}\in\mathbb{R}^I$ and a velocity vector $\mathbf{v}_p^{(t)}\in\mathbb{R}^I$.
Each particle further tracks its personal best position $\mathbf{pbest}_p^{(t)}$,
while the global best position among all particles is denoted by $\mathbf{gbest}^{(t)}$. 
For each candidate PA configuration $\mathbf{x}_p^{(t)}$, the fitness value is evaluated using the
ZF-based beamformer computed by \eqref{eq:zf_precoder}-\eqref{eq:zf_power_norm}.
The fitness function is defined as
\begin{equation}
f(\mathbf{x})
\triangleq
\begin{cases}
\displaystyle
\sum_{k\in\mathcal{K}}\log_2\!\bigl(1+\mathrm{SINR}_k\bigr),
& \text{if } \mathrm{SINR}_k \ge \mathrm{SINR}_k^{\min},~\forall k, \\[0.6em]
-\infty, & \text{otherwise}.
\end{cases}
\label{eq:fitness_zf}
\end{equation}
This directly embedds the SINR constraints into the search process, thus  
guiding the population toward feasible PA configurations with ZF-based beamforming.
If any user violates the minimum SINR requirement, the corresponding fitness value is set to $-\infty$,
such that infeasible PA configurations are automatically excluded from the search process.

\begin{algorithm}[!t]
    \caption{PSO-ZF for Joint Multi-PA Position and Digital Beamforming Optimization}
    \label{alg:pso_zf_multiPA}
    \begin{algorithmic}[1]
    \Require PSO parameters $\delta$, $n_{\rm p},\psi,u_1,u_2,v_{\max}$. 
    \State Initialize feasible PA positions $\mathbf{x}_p^{(0)}$ for each particle $p$. 
    \State Set $\mathbf{gbest} \gets \mathbf{x}_1^{(0)}$ and $f_{\rm gbest}\gets -\infty$.
    \For{$t=0,1,\ldots,T-1$}
        \For{$p=1,\ldots,n_{\rm p}$}
            \State Compute ZF precoder $\mathbf{W}(\mathbf{x}_p^{(t)})$ using \eqref{eq:zf_precoder}--\eqref{eq:zf_power_norm}.
            \State Evaluate fitness $f(\mathbf{x}_p^{(t)})$ by \eqref{eq:fitness_zf}.
            \If{$f(\mathbf{x}_p^{(t)}) > f_{\rm pbest,p}$}
                \State Update $\mathbf{pbest}_p \gets \mathbf{x}_p^{(t)}$, $f_{\rm pbest,p}\gets f(\mathbf{x}_p^{(t)})$.
            \EndIf
            \If{$f(\mathbf{x}_p^{(t)}) > f_{\rm gbest}$}
                \State Update $\mathbf{gbest} \gets \mathbf{x}_p^{(t)}$, $f_{\rm gbest}\gets f(\mathbf{x}_p^{(t)})$.
            \EndIf
        \EndFor
        \For{$p=1,\ldots,n_{\rm p}$}
            \State Sample $\mathbf{r}_{1,p}^{(t)},\mathbf{r}_{2,p}^{(t)} \sim \mathcal{U}([0,1]^I)$.
            \State Update $\mathbf{v}_p^{(t+1)}$ by \eqref{eq:pso_velocity_update} and clip $\mathbf{v}_p^{(t+1)}$ by $v_{\max}$. 
            \State Update PA position $\mathbf{x}_p^{(t+1)}$ by \eqref{eq:PA_position_pso}.
            \State Project $\mathbf{x}_p^{(t+1)} \gets \Pi_{\mathcal{X}}\big(\mathbf{x}_p^{(t+1)}\big)$ onto $\mathcal{X}$.
            \State Clamp by $\mathbf{x}_p^{(t+1)} \gets \Pi_{\mathcal{X}}\big(\mathcal{C}_{\delta}\big(\mathbf{x}_p^{(t+1)},\mathbf{x}^{(0)}\big)\big)$.
        \EndFor
    \EndFor
    \State \textbf{Return} $\mathbf{x}^\star \gets \mathbf{gbest}$ and $\mathbf{W}(\mathbf{x}^\star)$.
    \end{algorithmic}
\end{algorithm}

\subsubsection{Velocity and position updates}
The velocity of particle $p$ is updated according to
\begin{equation}
    \begin{split}
        \mathbf{v}_p^{(t+1)}
        &=
        \psi\,\mathbf{v}_p^{(t)}
        +
        u_1 \mathbf{r}_{1,p}^{(t)} \odot
        \big(
        \mathbf{pbest}_p^{(t)}-\mathbf{x}_p^{(t)}
        \big)
        \\& +
        u_2 \mathbf{r}_{2,p}^{(t)} \odot
        \big(
        \mathbf{gbest}^{(t)}-\mathbf{x}_p^{(t)}
        \big),
    \end{split}
\label{eq:pso_velocity_update}
\end{equation}
where $\psi$ denotes the inertia weight, $u_1$ and $u_2$ are the cognitive and social coefficients,
respectively, $\mathbf{r}_{1,p}^{(t)}$ and $\mathbf{r}_{2,p}^{(t)}$ are independent random vectors
with i.i.d. entries uniformly distributed in $[0,1]$, and $\odot$ denotes element-wise multiplication.
The inertia weight $\psi$ controls how much of the previous velocity is retained, which balances exploration and exploitation. 
A larger $\psi$ encourages broader search (larger momentum), while a smaller $\psi$ damps the motion and favors local refinement. 
Moreover, the cognitive and social terms guide each particle toward its own best configuration
and the globally best configuration discovered by the swarm.

To prevent excessively large position updates, each entry of $\mathbf{v}_p^{(t+1)}$ is clipped to
a predefined maximum magnitude $v_{\max}$.
The PA positions are then updated as
\begin{equation}\label{eq:PA_position_pso}
\mathbf{x}_p^{(t+1)}=\mathbf{x}_p^{(t)}+\mathbf{v}_p^{(t+1)}.
\end{equation}
After each unconstrained update, a projection operator $\Pi_{\mathcal{X}}(\cdot)$ is applied
to enforce physical constraints on the ordering, boundary, and minimum spacing of PAs. 
The updated position vector is first projected onto a feasible set to restore
ordering, boundary, and minimum spacing constraints.
Then, the trust-region constraint is enforced by clamping each PA position to
$[x_{n}^{(0)}-\delta,\ x_{n}^{(0)}+\delta]$, followed by another projection onto $\mathcal{X}$ 
to ensure that all particles remain physically feasible
and within the prescribed local neighborhood.

The above steps are repeated for a fixed number of outer iterations, as summarized by Algorithm \ref{alg:pso_zf_multiPA}.
The best PA configuration and the corresponding baseband beamformer found by the swarm is selected as the final solution.
The developed algorithm enables fast fitness evaluation for PSO using ZF beamforming, while implicitly
accounting for constructive interference of the desired signals, inter-user interference suppression,
and the impact of guided-mode power leakage.


\section{Numerical Results}

This section presents numerical results for the two-user two-mode PASS.
Unless otherwise stated, we adopt the following simulation setting.
The main waveguide is an open dielectric rectangular waveguide with a rectangular-strip cross-section of height $7$~mm and width $5$~mm.
The core relative permittivity is $\varepsilon_r=4.0$ with refractive index $n_{\mathrm{core}}=2.0$, and the cladding is air with $n_{\mathrm{clad}}=1$.
At the operating frequency $f_c=28$~GHz, we consider the two lowest-order guided modes of the open dielectric rectangular waveguide, namely the fundamental quasi-TE$_0$ mode and the first higher-order quasi-TE$_1$ mode.
Since guided modes in open dielectric waveguides are not strictly TE/TM, the modes are referred to as quasi-TE according to their dominant transverse electric-field polarization. 
The corresponding propagation constants $\beta_m$ are computed using an eigenmode solver, leading to $n_{\mathrm{eff},1}=1.7036$ and $n_{\mathrm{eff},2}=1.0892$.
Each PA is implemented as a short, parallel single-mode dielectric-waveguide segment with coupling length $L=6$ mm, placed in the near field of the main waveguide to enable non-contact evanescent coupling.
The propagation constant of each PA group is tailored to achieve phase matching at the center frequency.
In practice, $\beta_{n}^{\mathrm{PA}}$ is tuned by adjusting the auxiliary-waveguide dispersion via its cross-sectional geometry and dielectric loading, thereby enabling selective and independent coupling to the quasi-TE$_{0}$ and quasi-TE$_{1}$ modes, respectively.


\begin{figure}[!t]
    \centering
    \includegraphics[width=1\linewidth]{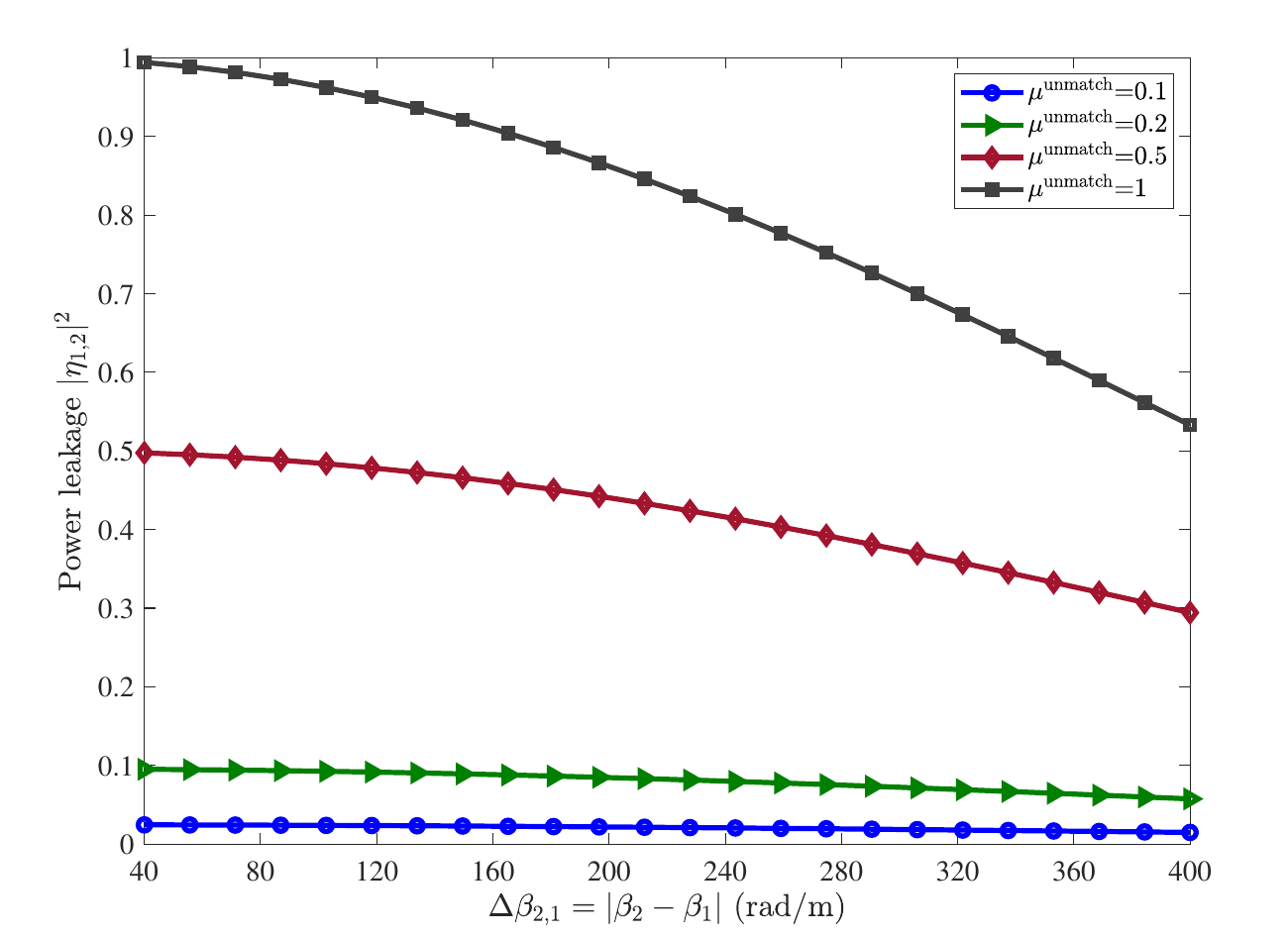}
    \caption{Power leakage versus $\Delta \beta_{2,1}$.}\label{fig:eta_beta}
\end{figure}

The waveguide length is set to $L_{\mathrm{wg}}=20$~m, and the waveguide is placed at a fixed height $h_{\mathrm{PA}}=2.5$~m above the user plane.
Along the waveguide, each PA $n$ can slide within $x_n\in[0,L_{\mathrm{wg}}]$ with a minimum inter-PA spacing $\Delta_{\min}=\lambda_0/2$.
Each user $k$ is located at $(a_k,y_k,0)$, where $a_k$ is uniformly distributed in $[3,L_{\mathrm{wg}}]$~m and $y_k$ is uniformly distributed in $[3,10]$~m.
The maximum transmit power and noise power are set to $P_{\max}=27$~dBm and $\sigma^2=-94$~dBm, respectively, and the minimum rate requirement is $1$~bps/Hz.
All numerical results are averaged over $100$ independent Monte Carlo realizations. 
In the two-PA case, Algorithm~\ref{alg:global_opt_refined} is implemented with a coarse grid size $N_x=1000$, Newton tolerance $10^{-9}$, a maximum of $40$ Newton iterations, and a bisection tolerance of $10^{-9}$.
In the multi-PA case, the PSO-ZF method in Algorithm~\ref{alg:pso_zf_multiPA} employs $n_{\mathrm{p}}=100$ particles over $T=50$ iterations, 
with maximum velocity $v_{\max}=5$, cognitive coefficient $u_1=1.4$, and social coefficient $u_2=1.2$.
The following baselines are considered for comparisons.
\begin{itemize}
    \item \textbf{Single-mode PASS with TDMA}: Only the fundamental mode is excited by a single feed in the waveguide. 
    Hence, two users are served in a TDMA manner, and each user occupies  $1/2$ time fraction. 
    In each time fraction, the PAs are scheduled to maximize the SNR \cite{Xu2025Rate}, with average power radiation among PAs. 
    \item \textbf{MISO with fully-digital beamforming (WMMSE)}: 
    An fixed-location array of $I=2$ antenna elements and two RF chains is deployed at the same position of the BS at $(0,0,h_{\mathrm{PA}})$. 
    The baseband beamforming is optimized by WMMSE to maximize the downlink sum rate, 
    serving as a benchmark for spatial multiplexing in multi-user transmission 
    but without the large-scale path-loss control. 
    \item \textbf{MISO with hybrid beamforming}: 
    The fixed-location MISO array utilizes a hybrid beamforming structure 
    with two RF chains and $I>2$ antenna elements through a phase-shifter network.
    A dual-loop hybrid beamforming algorithm is exploited.
    The outer loop computes a target fully-digital WMMSE precoder, 
    whereas the inner loop approximates it by alternatively updating the digital beamformer via least squares 
    and the constant-modulus analog beamformer via phase-only projection.
\end{itemize}

Fig.~\ref{fig:eta_beta} demonstrates that the power leakage decreases monotonically as $|\Delta\beta_{2,1}|$ increases. 
Moreover, a negligible power leakage can be guaranteed when the field selectivity $\mu^{\mathrm{unmatch}}=0$ for unmatched mode and PA is small, 
which can be realized through engineering designs of PAs. 
We use the following setting in the sequel. 
For non-leakage regime, the power leakage from the guided mode to unmatched PAs are ignored with $\mu^{\mathrm{unmatch}}=0$.  
For weak-leakage regime, a weak power leakage arises from the guided mode to unmatched PAs with $\mu^{\mathrm{unmatch}}=0.5$, 
as indicated in Fig.~\ref{fig:eta_beta}. 

\begin{figure}[!t]
    \centering
    \includegraphics[width=1\linewidth]{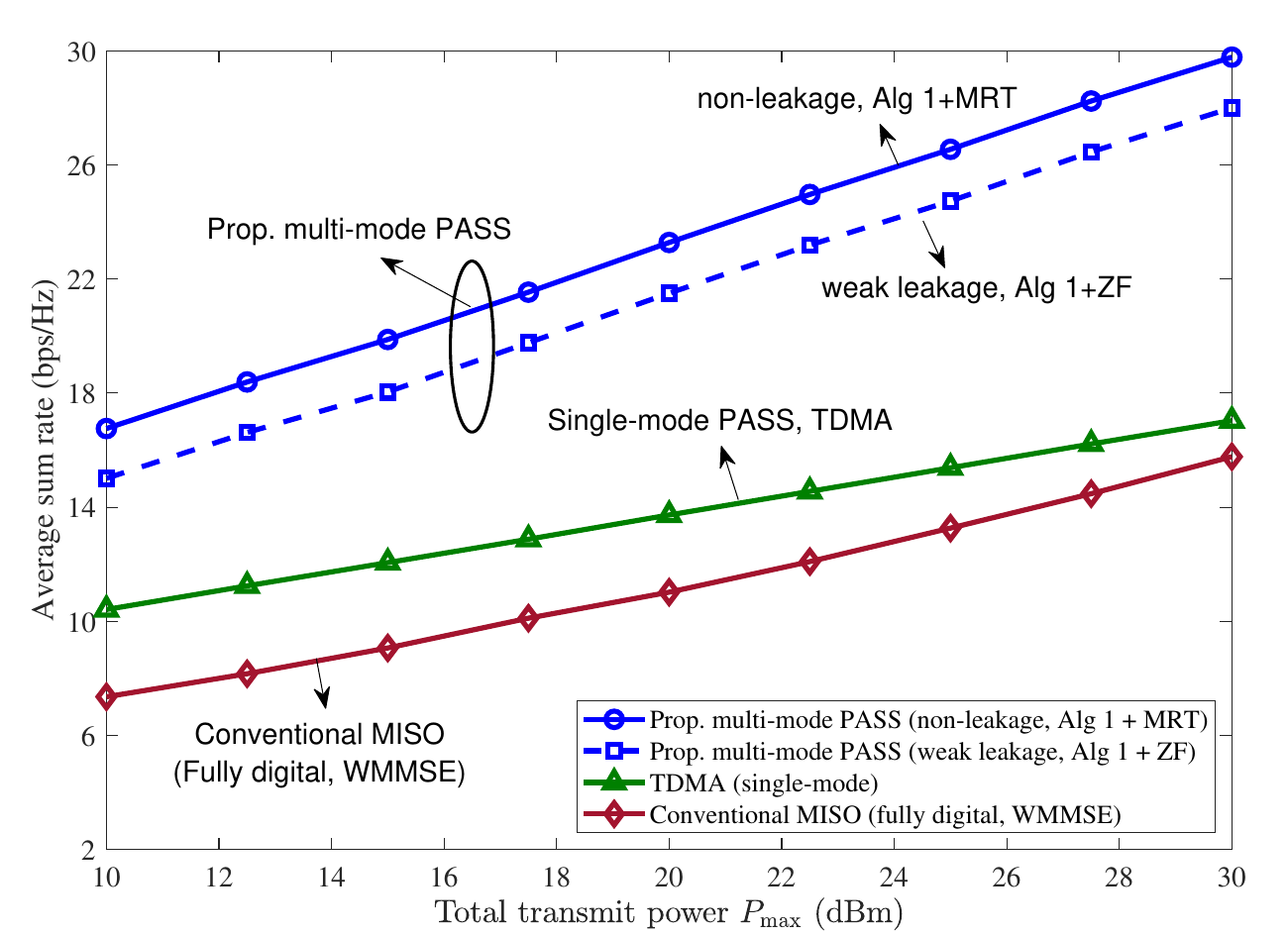}
    \caption{Sum rate under different $P_{\max}$.}\label{fig:rate_Pmax}
\end{figure}

\begin{figure}[!t]
    \centering
    \includegraphics[width=1\linewidth]{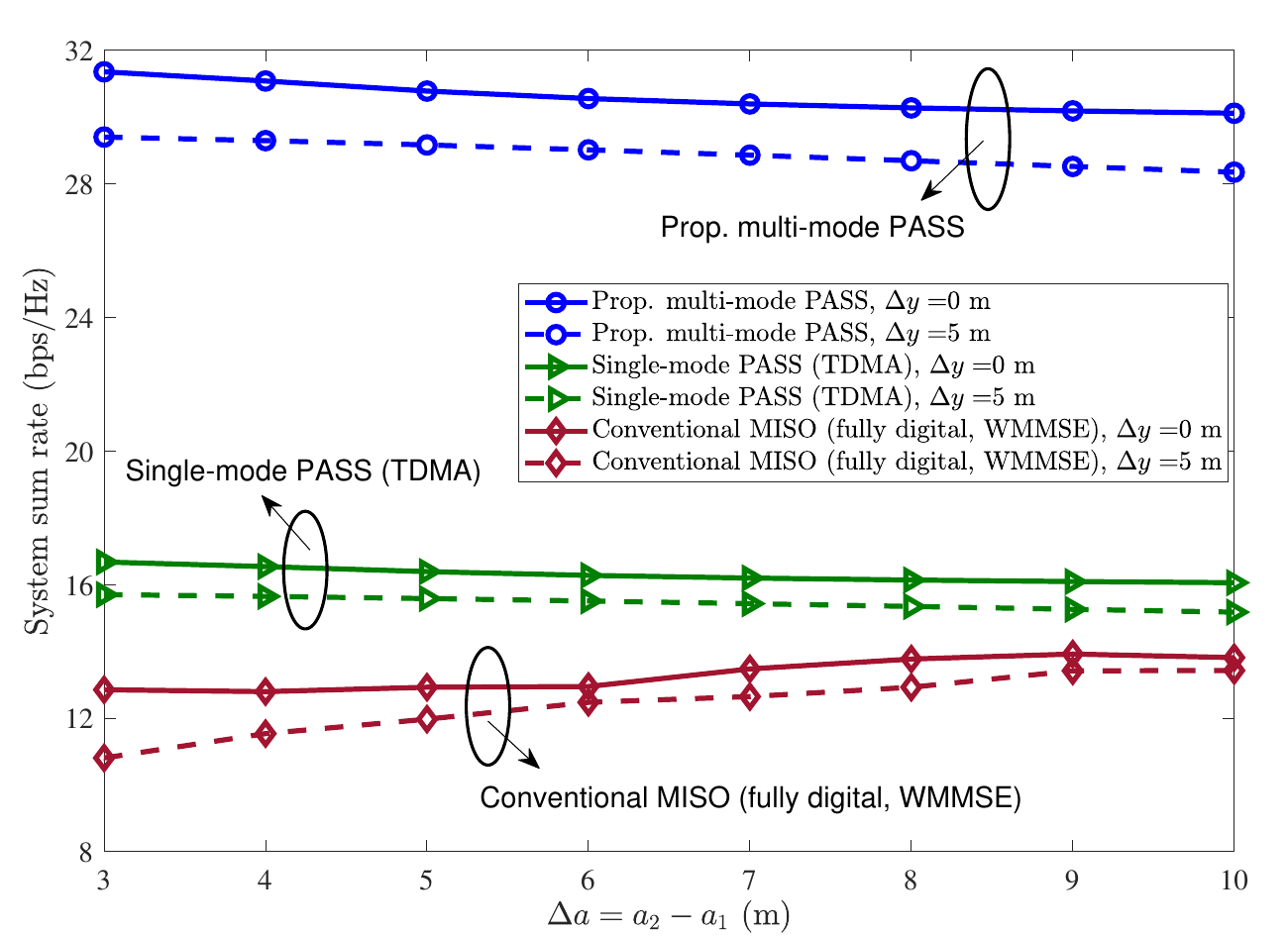}
    \caption{Sum rate comparisons under different user distributions.}\label{fig:rate_user_distribution}
\end{figure}

\subsection{System Performance for Two-PA Case}
We evaluate the system performance for two-PA case with $I=2$, where each PA is matched to one guided mode (i.e., $|{\mathcal I}_1|=|{\mathcal I}_2|=1$). 
For the non-leakage regime, the optimal PA positions is achieved by Algorithm~\ref{alg:global_opt_refined} with MRT beamforming. 
For the weak-leakage regime, the ZF beamforming is alternatively exploited with Algorithm~\ref{alg:global_opt_refined} to cancel the residual interference induced by the power leakage.

Fig.~\ref{fig:rate_Pmax} compares the system sum rate under different $P_{\max}$.
As shown in Fig.~\ref{fig:rate_Pmax}, the proposed multi-mode PASS achieves the highest sum rate by simultaneously serving users through mode-domain multiplexing. 
This demonstrates that the optimized PA positions can satisfy the two-scale channel orthogonality equalities for interference nulling, 
while reducing large-scale path losses. 
Notably, the proposed multi-mode PASS maintains the superiority even in the weak-leakage setting. 
This implies that effective mode-domain multiplexing can be achieved as the multi-mode PASS enables a full-rank equivalent PASS channel. 
In contrast, the single-mode PASS with TDMA suffers from a performance degradation because each user occupies only half of the resource block, leading to inefficient spectral utilization. 
While WMMSE-based fully digital beamforming is adopted, the fixed-antenna MISO leads to the worst performance owing to the large path losses experienced by the users. 

Fig.~\ref{fig:rate_user_distribution} further compares system sum rates under different user distributions, where we configure $a_2 > a_1$ and $y_2 > y_1$. 
As users move farther away along the $y$-axis direction, 
all schemes experience slight rate degradation due to increased path loss. 
When user distribution becomes more dispersed, i.e., $\Delta a = a_2 - a_1$ increases, the proposed multi-mode PASS experiences a slight performance loss, 
since the interference nulling may compromise the path gains. 
In comparison, the performance of the single-mode PASS remains unchanged, since PAs are always placed in the nearst locations from users.  
Furthermore, the sum rate acheived by the fixed-antenna MISO system increases as users become more dispersed, because their 
spatial channel correlations are decreased, and the path loss of the near-end user is reduced. 
Nevertheless, the proposed multi-mode PASS maintains a consistent performace gain over the TDMA baseline and the conventional MISO benchmark, 
indicating its effectiveness across various user distributions. 

\begin{figure}[!t]
    \centering
    \includegraphics[width=1\linewidth]{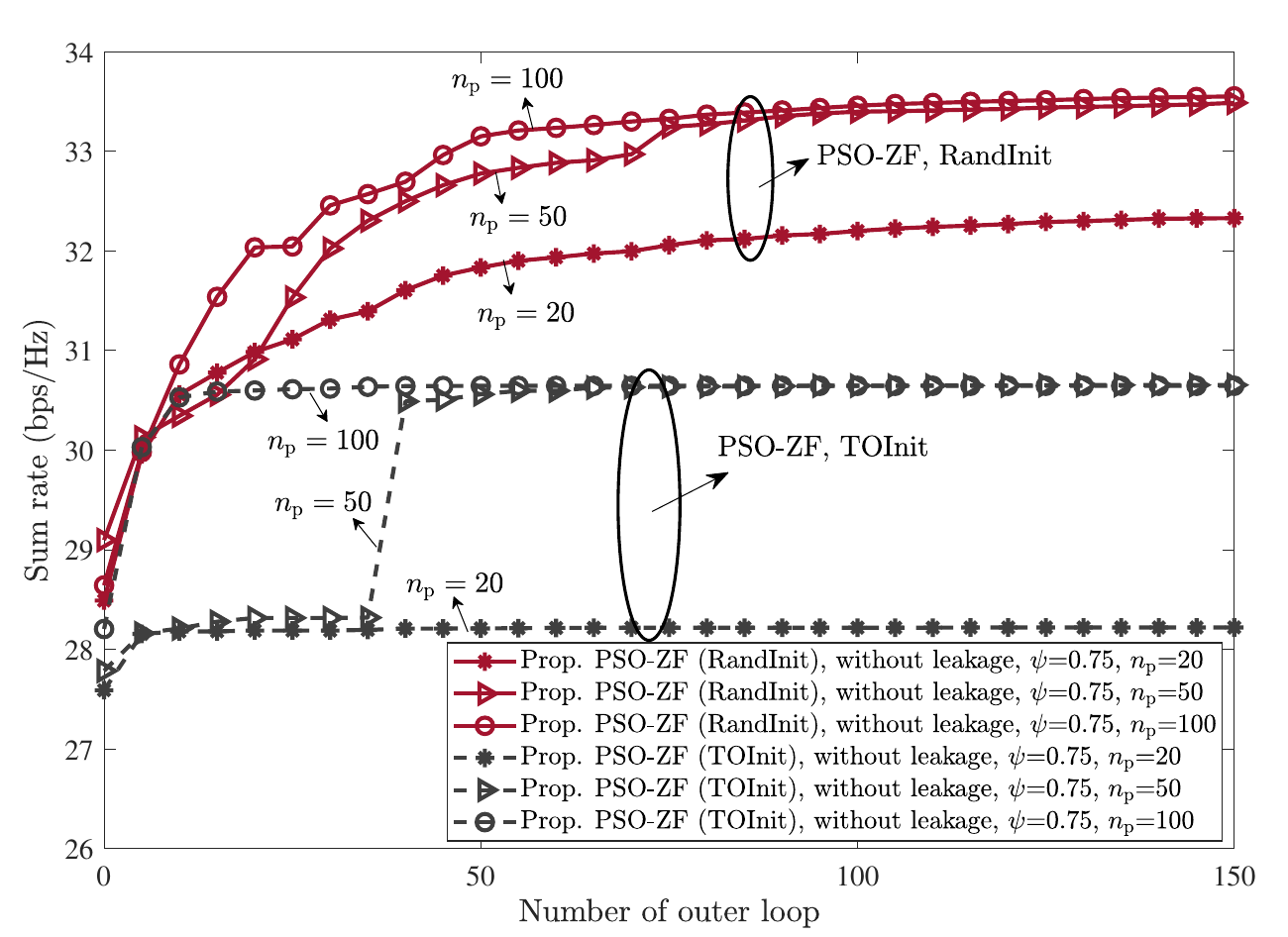}
    \caption{Convergence behaviours in non-leakage regime.}\label{fig:pso_zf_conv1}
\end{figure}

\begin{figure}[!t]
    \centering
    \includegraphics[width=1\linewidth]{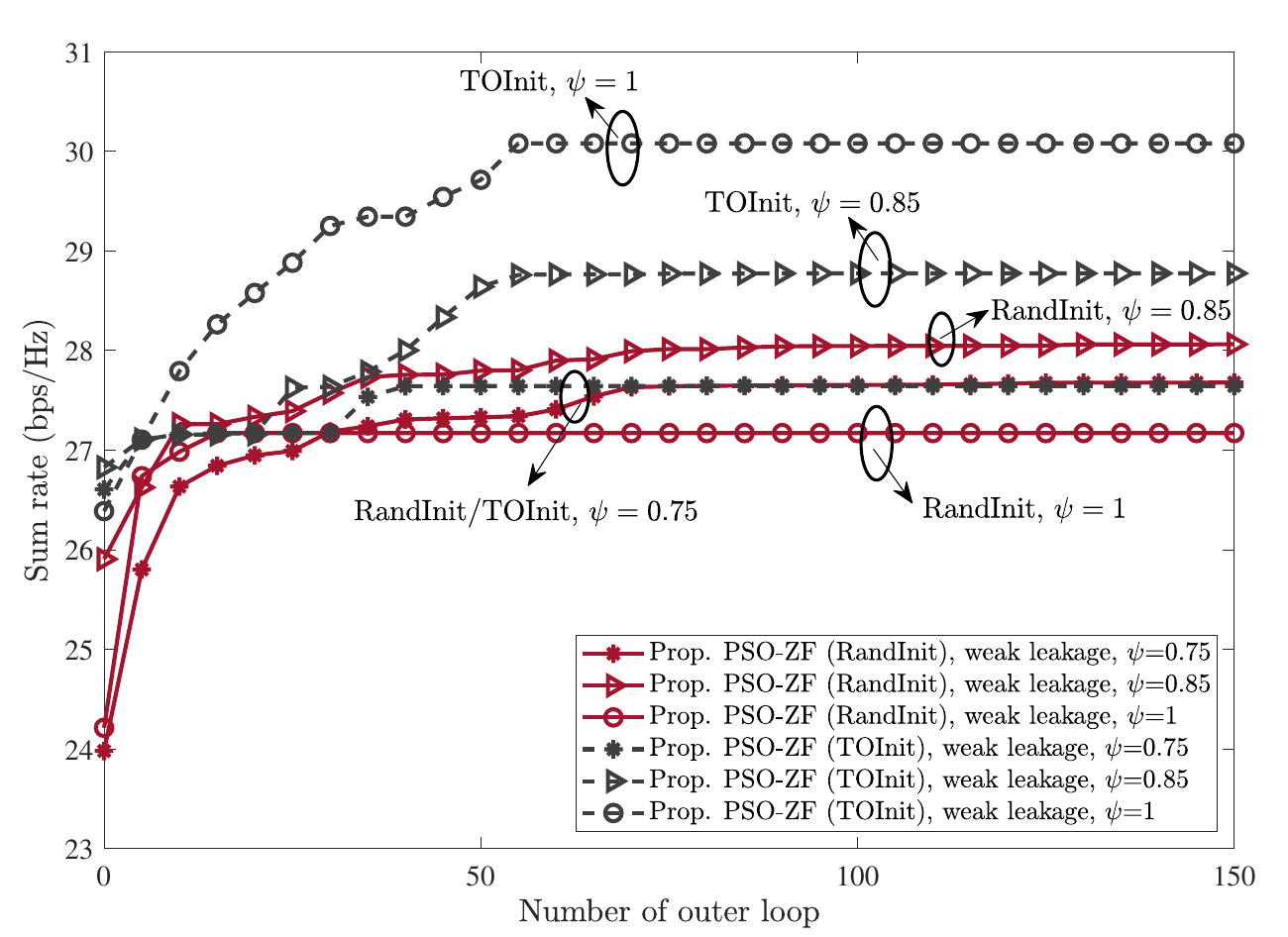}
    \caption{Convergence behaviours in weak-leakage regime.}\label{fig:pso_zf_conv2}
\end{figure}

\subsection{System Performance for Multi-PA Case}
In this part, we evaluate performance for the general multi-PA case using the proposed PSO-ZF in \textbf{Algorithm \ref{alg:pso_zf_multiPA}}. 
PSO particles are initialized around $\mathbf{x}^{(0)}$ with a small perturbation.
The following initialization schemes for $\mathbf{x}^{(0)}$ are considered: 
\begin{itemize}
    \item \textbf{Random Initialization (RandInit):}
    The PA position vector $\mathbf{x}^{(0)}$ is randomly initialized within $[x_{\min},x_{\max}]$, and is then projected to satisfy the spacing constraint. 
    \item \textbf{Two-PA-Optimum based Initialization (TOInit):}
    $\mathbf{x}^{(0)}$ is initialized by the optimum $(x_1^\star,x_2^\star)$ in two-PA case obtained by \textbf{Algorithm~\ref{alg:global_opt_refined}}. 
    The PAs in ${\mathcal I}_1$ and ${\mathcal I}_2$ are placed as tightly packed clusters centered at $x_1^\star$ and $x_2^\star$ with spacing $\Delta_{\min}$, respectively.
    This initializer provides a topology-aware warm start.
\end{itemize}

Fig.~\ref{fig:pso_zf_conv1} illustrates the convergence behaviours in the non-leakage regime (i.e., $\eta_{n,m}=0$ for $i\notin{\mathcal I}_m$). 
As shown in Fig.~\ref{fig:pso_zf_conv1}, increasing the particle number $n_{\rm p}$ improves both the convergence speed and the final sum rate under RandInit scheme. 
In particular, $n_{\rm p}=100$ achieves the best performance and saturates around the highest plateau, while $n_{\rm p}=20$ converges to a noticeably lower value. 
This is because a larger swarm provides stronger exploration capability in the highly nonconvex search space $\mathcal{X}$.
Moreover, TOInit exhibits a much earlier stabilization but converges to a lower system sum rate than RandInit in this non-leakage setting. 
This indicates that, when unmatched leakage is absent, 
the globally favorable multi-PA configuration is not necessarily a pair of dense clusters around $(x_1^\star,x_2^\star)$. 
Instead, distributing PAs more flexibly along the waveguide can improve the conditioning of $\mathbf{H}_{\rm eff}(\mathbf{x})$ and thus enhance the ZF sum rate. 

\begin{figure}[!t]
    \centering
    \includegraphics[width=1\linewidth]{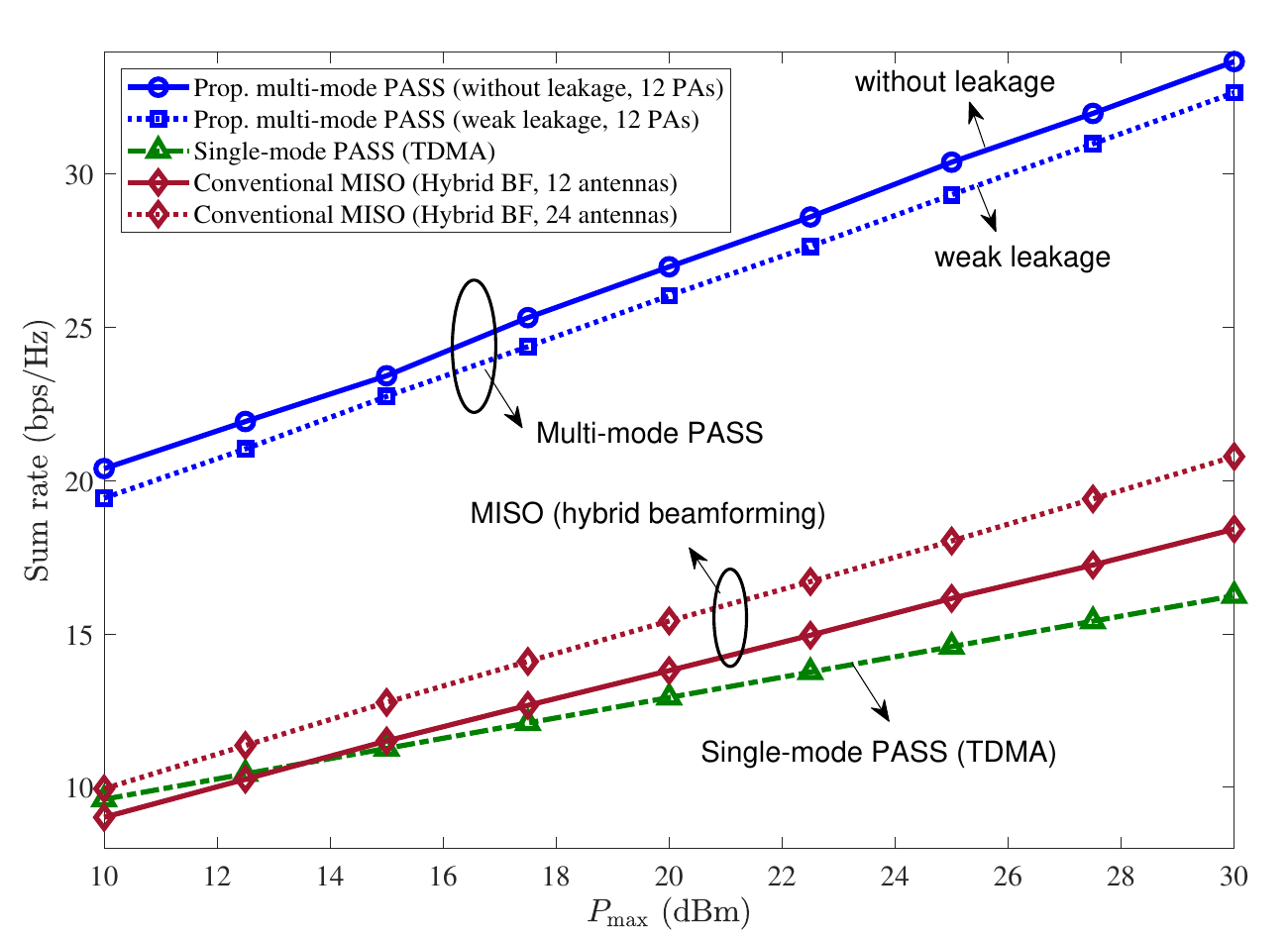}
    \caption{Sum rate versus $P_{\max}$ for multi-PA case.}\label{fig:rate_Pmax_multiPA}
\end{figure}

\begin{figure}[!t]
    \centering
    \includegraphics[width=1\linewidth]{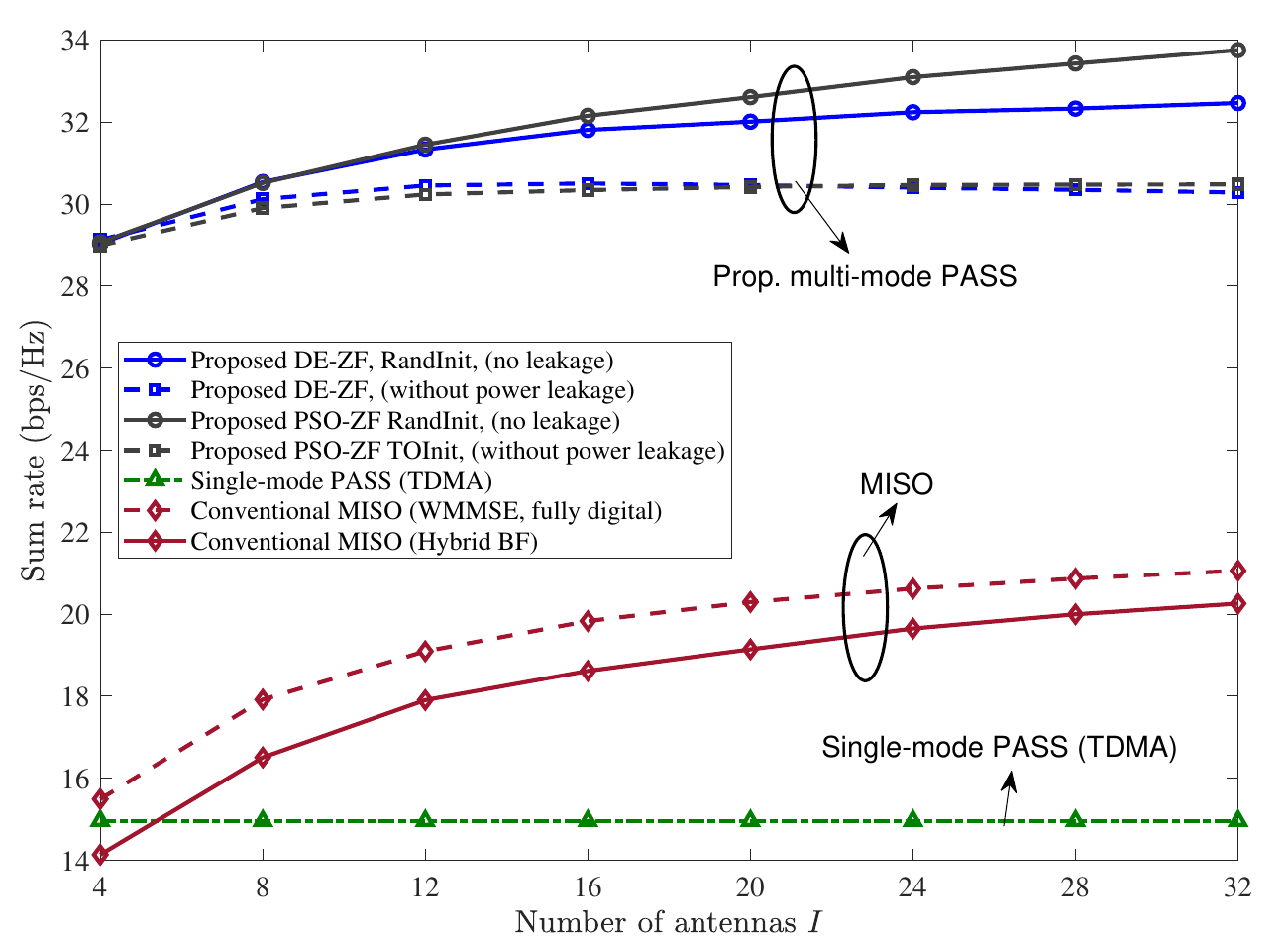}
    \caption{Sum rate versus number of antennas $I$.}\label{fig:rate_I_multiPA}
\end{figure}

Fig.~\ref{fig:pso_zf_conv2} further demonstrates the convergence behaviours in the weak-leakage regime, where different inertia weights $\psi$ are adopted for the PSO velocity update. 
In this regime, TOInit scheme with $\psi=1$ achieves the highest sum rate and converges reliably, whereas RandInit scheme yields a significantly lower performance. 
This may because under the power leakage, randomly distributed PAs may create strong residual interference during the search. 
The two-PA-based clustered placement provides a better-conditioned starting point that implicitly reduces harmful unmatched interactions and guides the swarm toward feasible high-rate regions.
Moreover, for RandInit scheme, a moderate $\psi$ (e.g., $\psi=0.85$) yields the best performance, 
while too large $\psi$ (e.g., $\psi=1$) may cause excessive velocity and oscillation around feasible regions, 
thus slowing down refinement. 
By contrast, TOInit scheme is more robust to larger $\psi$ because its initial structure 
already lies closer to a good feasible basin, and a larger inertia can help escape shallow local traps introduced by leakage.

Fig. \ref{fig:rate_Pmax_multiPA} displays the sum rate versus $P_{\max}$ with $I=12$ PAs. 
The proposed multi-mode PASS significantly outperforms TDMA-based single-mode PASS, confirming that it achieves substantial multiplexing gains. 
Compared with conventional MISO using hybrid beamforming and even increased antennas, 
the proposed PASS achieves a higher sum rate, highlighting the benefit of path loss control and mode-domain multiplexing. 
Moreover, the weak-leakage regime remains close performance to the non-leakage regime, demonstrating that the proposed algorithm is resilient to different settings.

Fig.~\ref{fig:rate_I_multiPA} shows the sum rate versus the number of PAs/antennas $I$. 
The baseline schemes involves another population-based optimization algorithm, i.e., differential evolution \cite{storn1997differential} with ZF beamforming (DE-ZF). 
The performance of both multi-mode PASS and MISO increases with $I$ since more antennas facilitate interference suppression, 
while the gain diminishes due to saturation. 
In contrast, performance of single-mode PASS is insensitive to $I$, as it relies on time-domain multiplexing, rather than spatial multiplexing. 
Moreover, PSO-ZF outperforms DE-ZF especially as $I$ increases, suggesting that PSO is more effective 
for the resulting high-dimension and high-oscillatory problem. 
The weak-leakage regime exhibits a slight performance loss due to undesired coupling, 
but the performance stays close to the non-leakage regime using the TOInit scheme.

\section{Conclusion}
This paper proposed a novel multi-mode PASS framework, 
where a single waveguide can simultaneously transmit multiple data streams via multiple modes, thus overcoming the DoF limitation of conventional PASS. 
A physic model was derived, which revealed the mode selectivity of PA power radiation. 
A two-mode PASS enabled two-user downlink communication system was examined. 
PAs were grouped to match with different modes for mode-selective signal radiation. 
PA locations and baseband beamforming were jointly optimized to maximize the sum rate under each user's minimum-rate constraints. 
For a non-leakage two-PA case, channel orthogonality was achieved by large-scale and wavelength-scale equalities on PA locations, 
which reduced the optimal beamforming to MRT. 
A Newton-based algorithm was then developed to efficiently search for the optimum. 
For a general multi-PA case, a low-complexity PSO-ZF method was proposed to tackle the high-oscillatory strong-coupled problem. 
Numerical results demonstrated the superiority of multi-mode PASS over conventional PASS and fixed-antenna MISO.

\appendices

\section{Proof of \textbf{Proposition \ref{prop:MS_two_mode}}}\label{proof:prop:MS_two_mode}

The proof follows coupled-mode analyses in the optical-waveguide and microwave literature (see, e.g.,
\cite[Sec.~13-3]{SnyderLove1983} and \cite[Sec.~9-7]{CollinFieldTheory}).
We adopt the following classical assumptions commonly used in weakly perturbed multi-mode waveguide theory
\cite[Sec.~13-3]{SnyderLove1983}, \cite[Ch.~7]{HausOpto}, which are typically satisfied in microwave and optical structures.
\begin{enumerate}
    \item[A1] \textit{(Modal separation)}
    Guided modes are non-degenerate and well separated in propagation constants, i.e.,
    $|\beta_m-\beta_{m'}| \gg \max_{j\in\mathcal{M}}|\kappa_j|$ holds $\forall m\neq m'$.
    This condition is typically satisfied for non-degenerate TE/TM modes in metallic and dielectric waveguides \cite{CollinFieldTheory}.

    \item[A2] \textit{(Weak inter-mode transfer)}
    The perturbation introduced by the PA does not induce noticeable energy transfer between guided modes, 
    i.e., $\left|\frac{\kappa_m \kappa_{m'}^\ast}{\beta_m-\beta_{m'}}\right|\ll |\kappa_m|$, $\forall m\neq m'$.
    Here, $|\kappa_m \kappa_{m'}^\ast|$ denotes the PA-mediated mode-to-mode coupling strength, whereas
    $|\beta_m-\beta_{m'}|$ quantifies the phase mismatch to suppress such transfer, which is generally satisfied under weak perturbations \cite{HardyStreiferOsinski1986}.
\end{enumerate}

\subsubsection{Reduction to $M$ Independent Two-Mode Subsystems}
We begin with the slowly varying envelope CMEs.
Define auxiliary envelope variables: 
\begin{align}\label{slowly_varying_envolope}
    a_m(\xi)\triangleq A_m(\xi)e^{+j\beta_m \xi},\qquad
    b(\xi)\triangleq B(\xi)e^{+j\beta^{\mathrm{PA}}\xi}.
\end{align}
Substituting \eqref{slowly_varying_envolope} into \eqref{eq:CME_env_A} and applying the chain rule yields
\begin{align}
    \frac{d a_m}{d\xi}
    &= e^{+j\beta_m\xi}\!\left(\frac{dA_m}{d\xi}+j\beta_m A_m\right)
     = -j\kappa_m B(\xi)e^{+j\beta_m\xi} \nonumber\\
    &= -j\kappa_m b(\xi)e^{j(\beta_m-\beta^{\mathrm{PA}})\xi},
    \qquad \forall m\in\mathcal{M}.
\end{align}
Applying the same transformation to $b(\xi)$, we have 
\begin{align}
\frac{d a_m}{d\xi}
&= -j\kappa_m\, b(\xi)\, e^{-j\Delta\beta_m \xi},
\qquad \forall m\in\mathcal M,
\label{eq:env_CME_a}\\
\frac{d b}{d\xi}
&= -j\sum_{m\in\mathcal M}\kappa_m^\ast\, a_m(\xi)\, e^{+j\Delta\beta_m \xi}.
\label{eq:env_CME_b}
\end{align}
For an arbitrary guided mode $i\in\mathcal{M}$, 
differentiating \eqref{eq:env_CME_a} gives
\begin{align}
\frac{d^2 a_i}{d\xi^2}
&= -j\kappa_i\frac{d}{d\xi}\!\left(b(\xi)e^{-j\Delta\beta_i\xi}\right)\nonumber\\
&= -j\kappa_i\left(
\frac{db}{d\xi}e^{-j\Delta\beta_i\xi}
-j\Delta\beta_i\, b(\xi)e^{-j\Delta\beta_i\xi}
\right).
\label{eq:an_second}
\end{align}
From \eqref{eq:env_CME_a}, we have $j\kappa_i b(\xi)e^{-j\Delta\beta_i\xi}= \frac{da_i}{d\xi}$.
Substituting \eqref{eq:env_CME_b} into \eqref{eq:an_second}, we can obtain
\begin{equation}
\frac{d^2 a_i}{d\xi^2}
\!+\! j\Delta\beta_i\frac{d a_i}{d\xi}
\!+\!|\kappa_i|^2 a_i(\xi)
\!+\!\sum_{m\neq i}\!
\kappa_i\kappa_m^\ast
a_m(\xi)e^{j\Delta\beta_{i,m}\xi}
=0,
\label{eq:coupling_structure}
\end{equation}
where $\Delta\beta_{i,m}\triangleq \Delta\beta_m-\Delta\beta_i=\beta_i-\beta_m$.
Equation \eqref{eq:coupling_structure} reveals two qualitatively different mechanisms.
The \emph{direct term} $|\kappa_i|^2a_i(\xi)$ captures the dominant first-order interaction between guided mode $i$ and the PA.
The remaining terms are PA-mediated \emph{cross interactions}, where energy transfers from guided mode $m$ to the PA and then to mode $i$ \cite{HardyStreiferOsinski1986}.

Under Assumption (A1), the phase-mismatch $\Delta\beta_{i,m}\triangleq \beta_i-\beta_m\neq 0$ makes
$e^{j\Delta\beta_{i,m}\xi}$ rapidly oscillatory. Let
$I_{i,m}(\xi)\triangleq \kappa_i\kappa_m^\ast a_m(\xi)$.
Applying a single integration-by-parts step yields 
\begin{align}
\int_0^L I_{i,m}(\xi)e^{j\Delta\beta_{i,m}\xi}d\xi
&=
\Big[\frac{I_{i,m}(\xi)}{j\Delta\beta_{i,m}}e^{j\Delta\beta_{i,m}\xi}\Big]_0^L
\\&-\frac{1}{j\Delta\beta_{i,m}}\int_0^L I'_{i,m}(\xi)e^{j\Delta\beta_{i,m}\xi}d\xi,
\nonumber
\end{align}
and hence
$
\Big| \int_0^L
\kappa_i\kappa_m^\ast a_m(\xi)e^{j(\beta_i-\beta_m)\xi}\,d\xi
\Big| \leq 
\Big(|a_m(0)|+ \allowbreak |a_m(L)|+\int_0^L\Big|\frac{d a_m}{d\xi}\Big| d\xi \Big) \frac{|\kappa_i\kappa_m|}{|\beta_i-\beta_m|}$. 
Under the slowly-varying-envelope (weak-coupling) regime, 
we have
\begin{equation}
    \left|
    \int_0^L
    \kappa_i\kappa_m^\ast a_m(\xi)e^{j(\beta_i-\beta_m)\xi}\,d\xi
    \right|
    \!=\!
    \mathcal{O}\!\left(\frac{|\kappa_i\kappa_m|}{|\beta_i-\beta_m|}\right),
    ~ m\neq n.
    \label{eq:cross_term_order}
\end{equation} 
From Assumption (A2), i.e., 
$\frac{|\kappa_i\kappa_m|}{|\beta_i-\beta_m|}\ll |\kappa_i|$ for $m\neq i$,
all cross-interaction terms are negligible in a first-order CMT approximation.
Hence, \eqref{eq:coupling_structure} decouples across $m$ in the weak-coupling regime, and the original multi-mode CME system approximately block-diagonalizes into
$M$ independent two-mode subsystems $(a_m(\xi),b_m(\xi))$, each governed by
\begin{align}
    \frac{d a_m}{d\xi}
    &= -j\kappa_m b_m(\xi)e^{-j\Delta\beta_m\xi},
    \label{eq:two_mode_a}\\
    \frac{d b_m}{d\xi}
    &= -j\kappa_m^\ast a_m(\xi)e^{+j\Delta\beta_m\xi},
    \label{eq:two_mode_b}
\end{align}
which is the canonical two-mode CMT model \cite[Sec.~13-3]{SnyderLove1983}.

\subsubsection{Closed-Form Expression}
Define symmetrized envelopes
\[
\tilde a_m(\xi)\triangleq a_m(\xi)e^{+j\Delta\beta_m\xi/2},\qquad
\tilde b_m(\xi)\triangleq b_m(\xi)e^{-j\Delta\beta_m\xi/2}.
\]
Then \eqref{eq:two_mode_a} and \eqref{eq:two_mode_b} can be written as
\[
\frac{d}{d\xi}
\begin{bmatrix}
\tilde a_m\\[0.3em]
\tilde b_m
\end{bmatrix}
=
-j
\begin{bmatrix}
\Delta\beta_m/2 & \kappa_m\\
\kappa_m^\ast & -\Delta\beta_m/2
\end{bmatrix}
\begin{bmatrix}
\tilde a_m\\[0.3em]
\tilde b_m
\end{bmatrix},
\]
whose eigenvalues are $\pm\phi_m$, $\phi_m\triangleq \sqrt{|\kappa_m|^2+(\Delta\beta_m/2)^2}$.
Solving over $\xi\in[0,L]$ with $b_m(0)=0$ yields \cite{SnyderLove1983} 
\begin{equation}
b(L)=\sum_{m\in\mathcal{M}} b_m(L)
=\sum_{m\in\mathcal{M}}\frac{\kappa_m}{\phi_m}
\sin(\phi_m L)\,
e^{-j\Delta\beta_m L/2}a_m(0), 
\label{eq:bm_solution}
\end{equation}
where $a_m(0)=c_m e^{-j\beta_m x}$ is the incident signal at the entrance of the PA.
Substituting into \eqref{eq:bm_solution} and absorbing $e^{-j\Delta\beta_m L/2}$ by
$\kappa_m \leftarrow \kappa_m e^{-j\Delta\beta_m L/2}$, we obtain
\begin{equation}
    b(L)
    = \sum_{m\in\mathcal{M}}
    \frac{\kappa_m}{\phi_m}\sin(\phi_m L) e^{-j\beta_m x_n} c_m.
\end{equation}
By definition, coefficient $\eta_m$ satisfies $b_m(L)=\eta_m e^{-j\beta_m x}c_m$. Hence,
$\eta_m
=
\frac{\kappa_m}{\phi_m}
\sin(\phi_m L).$
This completes the proof.

\section{Proof of Corollary~\ref{corolllary:stationary_D}}\label{proof:corolllary:stationary_D}
Let $L(x)\triangleq \log D(x)=\sum_{k=1}^2\frac12\log\big((x-a_k)^2+z_k^2\big)$. Since $D(x)>0$, we have $D'(x)=0 \Leftrightarrow L'(x)=0$.
The first-order derivative is given by
\begin{equation}
L'(x)=\frac{x-a_1}{(x-a_1)^2+z_1^2}+\frac{x-a_2}{(x-a_2)^2+z_2^2}.
\end{equation}
With $u=x-a_1\in(0,a_2-a_1)$ (so $x-a_2=u-(a_2-a_1)$), clearing denominators yields the cubic
\[
2u^3-3(a_2-a_1)u^2+((a_2-a_1)^2+z_1^2+z_2^2)u-(a_2-a_1)z_1^2=0.
\]
Applying the shift $u=t+\frac d2$ produces the cubic $t^3+pt+q=0$ with $(p,q)$ as in the corollary, 
whose discriminant is $\Delta_{D}=(\frac q2)^2+(\frac p3)^3$. Hence it has exactly one real root if $\Delta_{D}\ge0$ and three distinct real roots if $\Delta_{D}<0$. 
Substituting Cardano's formula for $\Delta_{D}\ge0$ and the trigonometric form for $\Delta_{D}<0$ into $x=a_1+\frac{a_2-a_1}{2}+t$ 
gives the stated $\tau_s$. This completes the proof. 
\qed

\bibliographystyle{IEEEtran}
\bibliography{reference}

\end{document}